\documentclass[aps,pre,superscriptaddress,showkeys,twocolumn,preprintnumbers]{revtex4-2}

\usepackage{amsmath}
\usepackage{amssymb}
\usepackage{xspace}
\usepackage{graphicx}
\usepackage{xcolor}
\usepackage{amsthm}

\newcommand{\Alphabet}{\mathcal{A}}
\newcommand{\FullShift}{\Alphabet^\mathbb{Z}}
\newcommand{\ShiftSpace} {\mathcal{X}}
\newcommand{\Forbidden} {\mathcal{F}}
\newcommand{\SubShift} {\mathsf{X}_\Forbidden}

\newcommand{\FullSTShift} {\MeasAlphabet ^{\mathbb{Z}^{d+1}}}
\newcommand{\Presentation}{\mathcal{P}}
\newcommand{\Language}{\mathcal{W}}

\newcommand{\msym}{x}

\newcommand{\MeasAlphabet}{\mathcal{A}}
\newcommand{\state}{\msym}

\newcommand{\site} {r\xspace}
\newcommand{\point} {(\site, t)\xspace}

\newcommand{\lattice} {\mathcal{L}\xspace}
\newcommand{\stfield} [2] {\ensuremath{\mathbf{\msym}_{#1} ^{#2}}\xspace}

\newcommand{\stpoint} {\stfield{t}{\site}}

\newcommand{\stpointprime} {\stfield{t'}{\site'}}

\newcommand{\V}{\mathtt{V}}

\newcommand{\radius} {R}
\newcommand{\neighborhood} {\eta}

\newcommand{\word} [3] {\ensuremath{{#1}_{#2:#3}}\xspace}

\newcommand{\plc}{\mathtt{L}^-}
\newcommand{\flc}{\mathtt{L}^+}

\newcommand{\causalfield} [2] {\ensuremath{\mathcal{S}_{#1} ^{#2}}\xspace}

\newcommand{\causalpoint} {\causalfield{t}{\site}}

\newcommand{\past}{\overleftarrow{x}}
\newcommand{\Past}{\overleftarrow{X}}
\newcommand{\future}{\overrightarrow{x}}
\newcommand{\Future}{\overrightarrow{X}}

\newcommand{\hmu}{h_\mu}
\newcommand{\eM}{$\epsilon$-machine}

\renewcommand{\H}{\operatorname{H}}

\newtheorem*{proposition}{Proposition}
\newtheorem*{theorem}{Theorem}

\begin{document}


\title{Algebraic Theory of Patterns as Generalized Symmetries}


\author{Adam Rupe}
\email[]{adamrupe@lanl.gov}
\affiliation{Center For Nonlinear Studies and Computational Earth Science, Los Alamos National Laboratory}

\author{James P. Crutchfield}
\email[]{chaos@ucdavis.edu}
\affiliation{Complexity Sciences Center, Department of Physics and Astronomy, University of California Davis}


\date{\today}

\begin{abstract}
We generalize the exact predictive regularity of symmetry groups to give an
algebraic theory of patterns, building from a core principle of future
equivalence. For topological patterns in fully-discrete one-dimensional
systems, future equivalence uniquely specifies a minimal semiautomaton. We
demonstrate how the latter and its semigroup algebra generalizes translation
symmetry to partial and hidden symmetries. This generalization is not as
straightforward as previously considered. Here, though, we clarify the
underlying challenges. A stochastic form of future equivalence, known as
predictive equivalence, captures distinct statistical patterns supported on
topological patterns. Finally, we show how local versions of future equivalence
can be used to capture patterns in spacetime. As common when moving to higher
dimensions, there is not a unique local approach, and we detail two local
representations that capture different aspects of spacetime patterns. A
previously-developed local spacetime variant of future equivalence captures
patterns as generalized symmetries in higher dimensions, but we show this
representation is not a faithful generator of its spacetime patterns. This
motivates us to introduce a local representation that is a faithful generator,
but we demonstrate that it no longer captures generalized spacetime symmetries.
Taken altogether, building on future equivalence, the theory defines and
quantifies patterns present in a wide range of classical field theories.
\end{abstract}

\keywords{Organization, Structure, Pattern, Computational mechanics, Spacetime,
Translation symmetry, Predictive equivalence}


\maketitle


\section{Patterns in Nature}
Symmetry plays a central role in fundamental physics. When we look out at the
world around us, on the human scale, however, there is a notable lack of exact
symmetries. Cows are not spherical, for instance. The disconnect between
physics at the fundamental level and the human scale is often described in
terms of \emph{spontaneous symmetry breaking}. How and why spontaneous symmetry
breaking occurs so ubiquitously in natural systems are interesting and
challenging questions, but not our concern here. Rather, we are interested in
the question of what structures result from broken symmetries. In particular, a
special case of spontaneous symmetry breaking is \emph{spontaneous
self-organization}. But what \emph{is} organization in the first place? Can we
mathematically formalize it and quantitatively measure it? 

We will use the general rubric \emph{pattern} to refer to the forms of
organization that spontaneously emerge in natural systems. When a system
undergoes a spontaneous symmetry breaking event it self-organizes into some
pattern, either in time, space, or both.

In many canonical examples of \emph{pattern formation} near equilibrium
\cite{Cros93a,Ball99a,Cros09a}---recall the convection roles that emerge in
B\'{e}nard flow \cite{Bena00a,Paul03a} or the spiral waves in the
Belousov-Zhabotinsky chemical reaction \cite{Zhab64a,Epst98a}---a system
undergoes a continuous-to-discrete symmetry breaking bifurcation event. This
occurs when it self-organizes, going from a homogeneous state with trivial
continuous spacetime symmetries to a state with nontrivial discrete symmetries.

Hexagonal convection cells form in a fluid with velocity initially zero
everywhere during a B\'{e}nard instability. Belousov first discovered a
``chemical clock'' with a discrete-time symmetry oscillation that arises from
an initially stationary mixture. Farther from equilibrium, these discrete
symmetries may be further broken during subsequent bifurcations, resulting in
states we may consider ``patterned'', but that have no discernible simple
symmetries remaining. Turbulent fluid flows containing coherent structures
\cite{Holm12a,Hall15a}, such as Jupiter's Great Red spot \cite{Marc88b,Somm88a},
provide many commonly-encountered examples. 

What we would like then, and what the following contributes, is a unified
account of patterns---an account that rigorously and formally describes the
full range of phenomenon including, but going beyond, organizations with exact
symmetries. Given that symmetries are formally captured using the mathematics
of group theory and given the enormous success group theory has had in
formulating physical theory \cite{Tung85a}, we take an algebraic approach to
framing patterns as generalized symmetries.

We start with the simplest setting of discrete one-dimensional spatial systems
(e.g., spin lattices) and show how the semigroup algebra of semiautomata---a
mathematical representation originating in symbolic computation---generalizes
translational symmetry. In doing so, we rigorously clarify subtleties posed by
this generalization---subtleties that have not been previously addressed. Based
on this, we introduce a classification hierarchy in terms of exact symmetries,
partial symmetries, hidden symmetries, and general patterns. We also describe
distinct statistical structures supported on these one-dimensional patterns and
show that stochastic generalizations of semiautomata provide mathematical
representations of these statistical structures.

In addition, we explore generalizations to patterns in higher dimensions.  We
introduce a class of local models that generalize the semiautomata approach for
spatiotemporal systems. Two models in this class are shown to be particularly
useful. Intriguingly, the utility of these models are apparently mutually
exclusive. The first model, introduced previously, can discover hidden
spacetime symmetries and coherent structures \cite{Rupe18a}, such as those in
turbulent fluid flows \cite{Rupe19a}. However, we describe a previously unknown
shortcoming of these models: they are not consistent generators of their
associated spacetime field patterns. The second model, introduced here for the
first time, corrects this flaw and introduces a consistent generator of
spacetime field patterns. Unfortunately, it loses the first model's useful
generalized spacetime symmetries.

These representations' conflicting strengths---generalized spacetime symmetries
and consistent spacetime generation---add new questions and suggest new paths
of inquiry to the enigma of patterns in higher dimensions
\cite{Schmi95a,Lind98a,Lind04a}.

\section{One-Dimensional Patterns}

Abstractly, we can think of a pattern as a \emph{predictive} regularity \cite{Shal98a}:
\begin{quote}
... some object $O$ has a pattern $P$---$O$ has a pattern `represented',
`described', `captured', and so on by $P$---if and only if we can use $P$ to
predict or compress $O$.
\end{quote}
On one extreme, symmetries represent an \emph{exact} predictive regularity. If
the symmetries are known, the pattern can be perfectly predicted at any other
point in time or space. On the other extreme, a completely random system is
entirely devoid of predictive regularity. If every point in spacetime is an
independent, identically distributed random variable, there is no regularity.
And so, knowledge of any part of the system cannot be leveraged to predict
other parts of the system. The notion of pattern that we seek encompasses both
of these extremes and systems in between. A general pattern will be neither
perfectly predicable nor entirely unpredictable---it will be an amalgamation
of regularity and randomness. 

Before proceeding, let us briefly compare and contrast the theory developed
here with the Pattern Theory of Ulf Grenander and colleagues \cite{gren96a}. As
both aim at a general quantitative understanding of what patterns are and
how to discover them in the world, there are many conceptual similarities.
While some quantitative similarities emerge, in particular the use of
nonparametric learning algorithms for hidden Markov models \cite{mumf10a}, most
of the quantitative machinery differs. Pattern Theory is grander in scope
than what is developed here, and so we do not require the very general
constructs of \emph{bonds}, \emph{connectors}, \emph{configurations},
\emph{images}, and the like \cite{gren96a}. Likewise, Pattern Theory does not
employ the machinery of symbolic dynamics, sofic shifts, and predictive
equivalence used here. Thus, our work is complementary to the more general
approach of Pattern Theory.

\subsection{Statistical Field Theories}

The following mainly concerns fully-discrete one-dimensional spatial systems.
These are given as a \emph{shift space} $\ShiftSpace$---a set of indexed
bi-infinite sequences, or \emph{strings}, of symbols taken from a finite
alphabet $\Alphabet$. Before diving into details, let's first take a moment to
compare shift spaces to the analogous setup from statistical mechanics for
analyzing ordered systems. 

A shift space can be thought of as a \emph{topological ensemble}---a set of
strings---in contrast with a statistical ensemble that is a distribution over a
set of strings. This is an abstraction of discrete-spin models in
one-dimension---e.g., $\Alphabet = \{-1, 1\}$ for a standard Ising model.
Rather than specify interactions on the spin lattice and analyze the resulting
statistical field theory, we wish to analyze any pattern present for a given
(topological) ensemble $\ShiftSpace$.

A key distinction between a shift space $\ShiftSpace$ as a topological ensemble
and a spin lattice ensemble in a statistical field theory is that all elements
$x \in \ShiftSpace$ are related to one another through the shift operator
$\sigma$. (Formally described below.) In fact, for the irreducible sofic shifts
we consider, $(\ShiftSpace, \sigma)$ is an ergodic dynamical system. And so,
every member $x$ of the ensemble can eventually be sampled through $\sigma$'s
action. Thus, we equivalently consider (i) $\ShiftSpace$ as an ensemble of
points and $\sigma$ as a deterministic mapping between those points or (ii)
$\ShiftSpace$ as a single infinite lattice and $\sigma$ moves indices on that
lattice. The difference is that of active versus passive transformation.

Spontaneous symmetry breaking in statistical field theories is monitored
through an \emph{order parameter}; such as total magnetization for an Ising
model. In the symmetric ``ordered'' phase, the order parameter has a nonzero
value and, after a symmetry-breaking phase transition, the order parameter
vanishes. For the Ising model, below the transition---below the \emph{critical
temperature}---spins tend to align giving nonzero magnetization. At zero
temperature the model reaches its ground state with all spins aligned. This is
a fully symmetric state with maximal magnetization, corresponding to strings
of the form $\{\ldots, 1,1,1,\dots\}$ or $\{\ldots, -1,-1,-1, \ldots\}$. Above
the critical temperature, this symmetry is fully broken, with zero
magnetization. 

While effective as an approach to thermal spin lattices, such as the Ising
and related lattice models, abstractly quantifying ``order'' with a single
scalar quantity---the order parameter---is far from ideal.

First, for the Ising model, there are only two configurations with maximal
magnetization, as given above. Second, these configurations are maximally
symmetric, with $\sigma^p(x) = x$ for integer $p$. Consider, though,
configurations of the form $\{\ldots, -1,1,-1,1,-1, \ldots\}$. These
configurations are still symmetric, with $\sigma^{2p}(x) =x $, although they
have vanishing order parameter. There are many such symmetric configurations
with zero magnetization: e.g., those of the form $\{\ldots, (-1)^n, (1)^n,
\ldots\}$, with $\sigma^{2np}(x) = x$. More novelly, there are zero
order-parameter configurations that are neither completely symmetric nor
completely random. Third, these symmetric sequences with zero order parameter
are not the ground state and they are not stable under thermal perturbations.
Thus, though singled out by the choice of total magnetization as \emph{the}
order parameter, they are edge cases that will almost never be seen. Finally
and more generally, order parameters in statistical mechanics are not
determined from first physical principles. They must be posited initially and
then proved appropriate.

Similarly, \emph{correlation functions} and \emph{structure factors} are
additional and commonly-employed scalar quantities that capture one or another
notion of order. Conceptually, a system considered highly ordered will surely
be highly correlated. Patterns that emerge on the macroscopic scale correspond
to collective behaviors on the microscale that certainly exhibit nonzero
correlations. However, as with order parameters, there can be many degeneracies
between specific patterns and correlation values. In short, order is something
beyond correlation. A diverging correlation length in an Ising model at the
critical temperature does not signify the presence of intricate patterns and
organization, such as spiral wave patterns in lattice models of excitable media
\cite{Green78a}. To remedy these failings we seek a definition of pattern that
is not scalar.

This is not to banish all scalar quantities. Many, in given settings, can be
insightful \cite{Crut97a,Feld98a,Feld02b,Robi11a}. We will show that the
algebraic \emph{presentations} for topological patterns have a natural
extension to patterns in statistical field theories. Moreover, scalar
quantities of interest, like correlation functions, can be computed in
closed-form from the stochastic presentations. We also return later, briefly,
to discuss generalized order parameters in light of the algebraic theory.

\subsection{Symbolic Dynamics}

We now detail shift spaces and how they quantify topological patterns as
generalized symmetries. Consider a finite alphabet of $n$ symbols $\Alphabet =
\{0, 1, \dots, (n-1)\}$ and (indexed) bi-infinite symbol sequences or
\emph{strings}. The set $\FullShift$ of all possible bi-infinite sequences is
known as the \emph{full-$n$ shift}. A particular sequence $x = \ldots x_{-1}
x_0 x_1 \ldots \in \FullShift$ is described as a \emph{point} in $\FullShift$.
For now, we need not specify whether sequence indices are time coordinates or
space coordinates. In either case, translations are generated via the
\emph{shift operator} $\sigma$ that maps a point $x \in \FullShift$ to another
point $y = \sigma(x)$ whose $i^{\text{th}}$ coordinate is $y_i = x_{i+1}$ for
all $i$. (That is, $\sigma$ shifts every element of $x$ one place to the left.)
Our interest is in patterns as predictive regularity, and the predictions are
made over translations generated by $\sigma$. Thus, we want to capture patterns
in closed, $\sigma$-invariant subsets of $\FullShift$. The subsets are called
\emph{shift spaces} (or \emph{subshifts} or simply \emph{shifts}).

Often one can concisely specify a shift space as the set of all strings that do
not contain a collection of \emph{forbidden words}. A \emph{word} is a finite
block of symbols $a_j \in \Alphabet$ and a point $x$ is said to contain or admit
a word $w = a_0 a_1 \cdots a_k$ if there are indices $i$ and $j = i + k$ such
that $\word{x}{i}{j} = w$; explicitly, $x_i = a_0, x_{i+1} = a_1, \ldots,
x_{j-1} = a_k$. Again, a word is a finite sequence of symbols; a string,
bi-infinite.

For a collection $\Forbidden$ of forbidden words, define $\SubShift$ to be the
subset of strings in $\FullShift$ that do \emph{not} contain any words $w \in
\Forbidden$. A shift space $\ShiftSpace$ is a subset of the full shift
$\FullShift$ such that $\ShiftSpace = \SubShift$ for some collection
$\Forbidden$ of forbidden words \cite{Lind95a}. The \emph{language}
$\Language(\ShiftSpace) = \Forbidden^c$ of a shift space $\ShiftSpace$ is the
collection of all words that occur in some point in $\ShiftSpace$.


If $\Forbidden$ is a finite set, the resulting shift space is called a
\emph{subshift of finite type} \cite{Smal67a} or an \emph{intrinsic} or a
\emph{topological Markov chain} \cite{Parr64a}. A wider class of
finitely-definable shift spaces are the \emph{sofic shifts}. These are the
closure of subshifts of finite type under continuous local
mappings---\emph{$k$-block factor maps} \cite{Weis73}. Though, note that there
are many equivalent definitions of sofic shifts; several of which are given
below as needed. A sofic shift is \emph{irreducible} if, for every ordered pair
of words $u, v \in \Language(\ShiftSpace)$, there is a word $w$ such that $uwv
\notin \Forbidden$. For reasons elaborated on shortly, we define general
discrete one-dimensional patterns as irreducible sofic shifts.


\subsection{Sofic Shifts as Topological Patterns}

To recap, we seek a mathematical specification of patterns in strings that
captures a range of organizations spanning fully-symmetric sequences to
arbitrary (``random'') sequences. Moreover, we wish to identify, from first
principles, an associated algebra that generalizes the group algebra of
symmetries. For physical consistency, we started with shift spaces since they
are shift-invariant subspaces of strings. To fulfill the algebraic requirement,
we now further restrict to sofic shifts as they are shift spaces defined in
terms of a finite semigroup \cite{Weis73}. 

Recall that a \emph{group} is a set of elements closed under an associative and
invertible binary operation with an identity element. In this, they are too
restrictive and impose only exact symmetry. In contrast, \emph{semigroups}
require neither invertibility nor an identity operation. This relaxation is key
to defining generalized patterns, as we laid out above. It permits exact
symmetries but also allows expressing noisy and approximate symmetries.

The elements of a sofic shift's semigroup are words and the binary operation is
word concatenation. For example, the set $\Language(\FullShift)$ of all words in
$\FullShift$ and their concatenation products together form the \emph{free
semigroup}. For example, the product of $u=00$ and $v=11$ in the free semigroup
is the word $w = uv = 0011$. A sofic shift $\ShiftSpace = \SubShift$ is defined
in terms of a finite semigroup $G$ with an absorbing element $e$, whose product
is $ge = eg =e$ for all $g \in G$. The absorbing element $e$ together with the
elements from the alphabet $\Alphabet = \{0, 1, \dots, (n-1)\}$ generates $G$
via single-symbol concatenations.  $G$'s production rules are such that for any
pair of allowed words $u, v \notin \Forbidden$, if their concatenation $w = uv$
contains a forbidden word $f \in \Forbidden$, then their product in $G$ gives
the absorbing element $uv = e$.


A semigroup of word concatenations is also associated with a simple
presentation in the form of a \emph{semiautomaton} finite-state machine
\cite{ginz68a}---the triple $(\Xi, \Alphabet, \mathcal{M})$, where $\Xi =
\{\xi_0, \xi_1, \dots, \xi_m\}$ is a set of \emph{internal states}, $\Alphabet
= \{0, 1, \dots, (n-1)\}$ is the symbol alphabet, and $\mathcal{M} = \{M_0,
M_1, \dots, M_{(n-1)}\}$ is a set of mappings from $\Xi$ into $\Xi$. To be
explicit, we consider \emph{deterministic and fully-specified
semiautomata} for which each $M_a$ is a function---each input has one and
only one output---with domain over the full set $\Xi$ of internal states.
Semiautomata can be usefully depicted as an edge-labeled directed graph.
The vertices represent the internal states in $\Xi$ and for every pair $(\xi_i,
\xi_j)$ such that $\xi_j = M_a(\xi_i)$ there is an edge labeled $a \in
\Alphabet$ that leads from $\xi_i$ to $\xi_j$. For a deterministic and
fully-specified semiautomaton, there is an edge labeled with each symbol in
$\Alphabet$ emanating from every state in $\Xi$. 

\begin{figure*}
\centering
\includegraphics[width = 0.8 \textwidth]{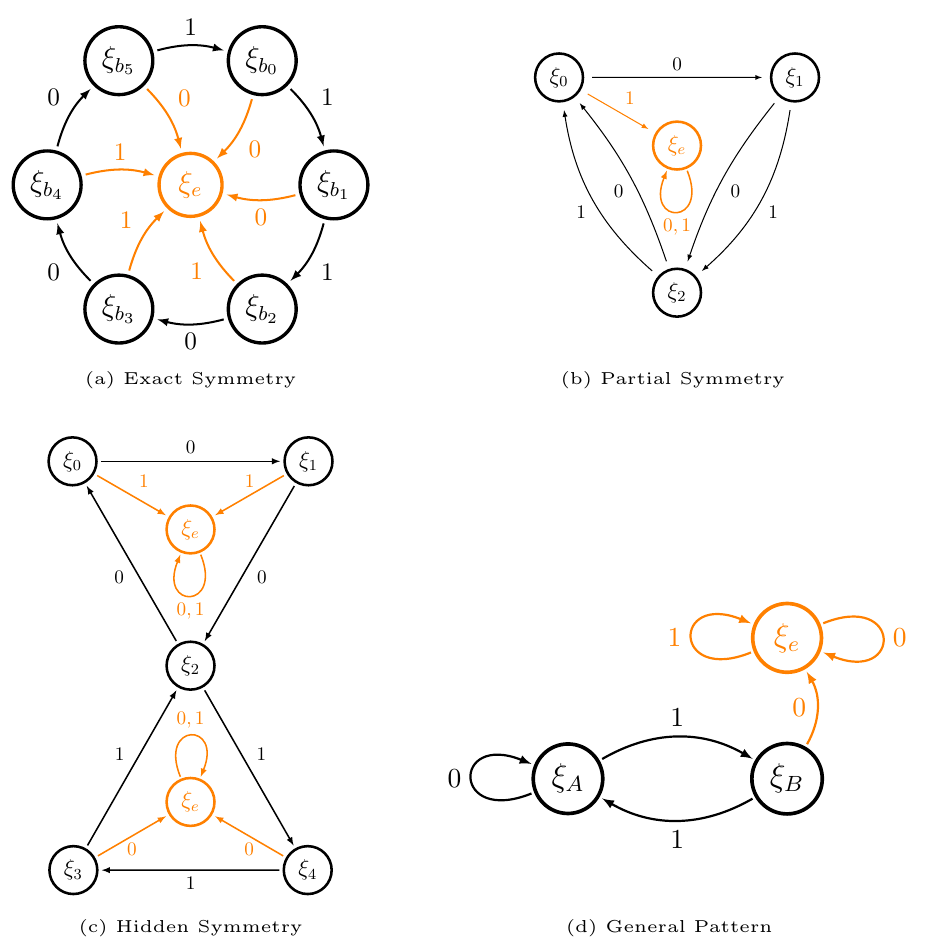}
\caption{Semiautomata presentations $\Presentation(\ShiftSpace)$:
	(a) Exact symmetry shift, (b) partial symmetry shift, (c) hidden symmetry
	shift, and (d) general pattern sofic shift.
	}
\label{fig:machines}
\end{figure*}

A fully-specified semiautomaton directly determines a subshift's algebra from
the free semigroup as follows. For every state $\xi_i \in \Xi$ and any element
$a_0a_1\cdots a_k$ of the the free semigroup $\FullShift$ there is a map
$M_{a_0} \circ M_{a_1} \circ \cdots \circ M_{a_k}$ from $\xi_i$ to another state
$\xi_j \in \Xi$ \cite{ginz68a}. A deterministic and fully-specified
semiautomaton is a \emph{presentation} of a sofic shift if we include an
absorbing ``forbidden'' state $\xi_e \in \Xi$. That is, the mappings associated
with all elements of the free semigroup containing a forbidden word in
$\Forbidden$ lead to the forbidden state. And, all mappings from the forbidden
state return the forbidden state. That is, $M_w(\xi_i) = \xi_e$ for all $\xi_i$
and $w \in \Forbidden$, and $M_a(\xi_e) = \xi_e$ for all $a \in \Alphabet$. See
Fig.~\ref{fig:machines} for presentations of example shifts.

Since sofic shifts are defined by a finite semigroup, every sofic shift can be
presented by a semiautomaton with a finite set of states $\Xi$.
Recall from above that the idea of \emph{compression} is related to our
intuition of pattern. A pattern---a predictive regularity---allows for a
compressed representation of a system's behaviors. Note that sofic shifts and
their presentations provide a finite representation of an ensemble of infinite
strings through their finite semigroup.

Here, we distinguish between three types of sofic-shift semiautomaton
presentation. Appendix~\ref{app:constructions} gives example constructions of
these three types of presentation.

The most straightforward presentation assigns a state $\xi \in \Xi$ to each
element of a semigroup $G$ of $\ShiftSpace$ and fills in the state transitions
$M_a$ using $G$'s production rules \cite{Kitc86a}. While straightforward to
construct, if a $G$ is known, this semiautomaton presentation is not
necessarily minimal, in terms of the state set size $|\Xi|$. To specify a
particular pattern as a sofic shift, it is crucial to have a minimal and unique
presentation associated with the given sofic shift. This also allows extracting
unambiguous quantitative measures of the ensemble of strings, such as measures
of correlation, from the finite presentation.

An important presentation that is minimal and constructible without knowing any
$G$ is $\ShiftSpace$'s \emph{future cover} \cite{Krie84a}, defined below. The
future cover semiautomaton of every irreducible sofic shift $\ShiftSpace$ has a
unique strongly connected component \cite{Fisc75a,Fisc75b}. This irreducible
component is our mathematical representation of patterns as generalized
symmetries. We refer to it as the \emph{canonical machine presentation
$\Presentation(\ShiftSpace)$ of $\ShiftSpace$}. The future cover and its
recurrent component $\Presentation(\ShiftSpace)$ provide a unique, minimal
mathematical representation of $\ShiftSpace$. 

The \emph{future set} (sometimes \emph{follower set}) $F_\ShiftSpace(w)$ of a
word $w \in \Language(\ShiftSpace)$ is the collection of all words $u$ such
that $wu \in \Language(\ShiftSpace)$. Define the \emph{future equivalence
relation} $\sim_F$ on $\ShiftSpace$ as:
\begin{align}
u \; \sim_F \; v \iff F_\ShiftSpace(u) = F_\ShiftSpace(v)
~,
\label{eq:future_equiv}
\end{align}
for all $u, v \in \Language(\ShiftSpace)$.

An important definition of sofic shifts that we use shortly is:
\begin{theorem}[1]
\cite[Theorem 3.2.10]{Lind95a}
A shift space is sofic if and only if it has a finite number of future sets.
\end{theorem}
Therefore, sofic shifts have finitely-many equivalence classes, denoted $[
\cdot ]_F$, and these equivalence classes plus the absorbing forbidden state
are the internal states $\Xi$ of the future cover semiautomaton.

The mappings $M_a$ that give the state transitions are defined from the allowed
concatenations that do not contain a forbidden word in $\Forbidden$. That is,
each state $\xi_i \in \Xi \setminus \{\xi_e\}$ is an equivalence class $[u]_F$
and, for each symbol $a \in \Alphabet$, the concatenation $v = ua \neq e$
belongs to the equivalence class $[v]_F$ assigned as state $\xi_j \in \Xi$,
giving $M_a(\xi_i) = \xi_j$. Note that this is independent of the choice $u \in
[u]_F$ and that $[v]_F$ in some cases may be equal to $[u]_F$. This gives a
self-edge transition in the semiautomaton: $M_a(\xi_i) = \xi_i$. If $ua = e$,
then $M_a (\xi_i)$ maps to the forbidden state $\xi_e$.

This natural transition structure follows from future equivalence and it leads
to a very important property of the future-cover semiautomata. They are
called \emph{unifilar} \cite{Ash65a} in information theory, equivalently also
called \emph{right resolving} in symbolic dynamics \cite[Corollary
3.3.19]{Lind95a} or \emph{deterministic} in automata theory \cite{Hopc06a}. A
fully-specified semiautomata is unifilar if for every internal state $\xi_i$
and every word (element of the free semigroup) $a_0 a_1\cdots a_k$ the map
$M_{a_0} \circ M_{a_1} \circ \cdots M_{a_k}$ leads to one and only one internal
state $\xi_j$. (It may be that $j$ = $i$.) 

Unifilarity is the defining property of a \emph{predictive} semiautomaton.
Since the goal is to formalize patterns as predictive regularities, this is an
important point to stress. By way of contrast, first note that any presentation
of a sofic shift $\ShiftSpace$, whether unifilar or not, is a \emph{generator}
of $\ShiftSpace$. Every string in $\ShiftSpace$ can be generated by following
the symbol-labeled transitions of the presentation and no forbidden words can
be generated. Thus, being generative can be thought of as the baseline property
of any model of a sofic shift. Prediction is an additional capability beyond
generation that arises from a presentation being unifilar. 

Unifilarity establishes a presentation as predictive in the following way. For
topological patterns, the task of prediction is to establish what may happen in
the future, given what has happened in the past. Specifically, given a word
$w$, what words are allowed to follow in the shift space? That is, what is
$w$'s future set? Due to the implied determinism, each internal state of a
unifilar presentation is uniquely specified by the word $w$ leading to that
state. Notably, this is not guaranteed for an arbitrary generator of a shift.
Furthermore, in a unifilar presentation every subsequent word again uniquely
leads to another internal state. It is straightforward to see that the future
cover is a predictive presentation: By definition, its internal states are
\emph{future separated}---the set of all words that may follow from each
internal state is unique. 

To establish that sofic shifts and their canonical machine presentations
express patterns as generalized symmetries, it is helpful to first describe how they capture exact translation symmetries of symbolic sequences.

\subsection{Exact Symmetries}

A string $x$ has a discrete translation symmetry if $\sigma^p(x) = x$, where
the minimal such $p \in \mathbb{N}$ is the symmetry's \emph{period}. The
symmetry group is the set $\{\sigma^{np}\ : \; n \in \mathbb{N}\}$ with
$\sigma^{ip} \sigma^{jp} = \sigma^{kp}$ where $k = i+j$. Since here $p$ is
finite, the action of $\sigma$ on $x$ produces a compact shift-invariant
subspace of $\FullShift$. Therefore, translation symmetric strings are shift
spaces $\ShiftSpace$. 

We now show that the shift space $\ShiftSpace$ of translation symmetric strings
is sofic. The action of the symmetry group used to define $x$ is determined by
the shift operator $\sigma$, while the action of sofic semigroups is word
concatenation.

To connect these, consider \emph{windows} $\word{x}{i}{j}$ ($i <
j)$ that return the word $w$ from coordinates $i$ through $j$ in $x$. For $(i <
j < k)$, if $\word{x}{i}{j} = u$ and $\word{x}{j}{k} = v$, then $\word{x}{i}{k}
= w =uv$ gives the concatenated window. Recall that $\sigma$ shifts indices in
$x$ so that $\sigma^k(x)_i = x_{i+k}$. If we have a word $u \in
\Language(\ShiftSpace)$, there is some $(i,j)$ such that $\word{x}{i}{j} = u$.
Then the allowed concatenations $uv \neq e$ are determined by the shift
operator $\sigma$, since we can write $v = \word{x}{j}{j+k} = (x_j) \sigma
(x)_j \sigma^2(x)_j \cdots \sigma^k (x)_j$. All such concatenations determine
$\ShiftSpace$'s semigroup $G$. For translation symmetric strings with
$\sigma^p(x) = x$, $\sigma^p(x)_i = x_i$. And so, intuitively, there is only a
finite number of elements in $G$. Therefore, $\ShiftSpace$ for translation
symmetric strings is sofic.

%

We now show this explicitly by constructing the canonical machine presentation
$\Presentation(\ShiftSpace)$ with a finite number of states. 

\begin{proposition}
Translation-symmetric strings, $\sigma^p(x) = x$ for some $p \in \mathbb{N}$, together with their shifts $y = \sigma^n(x)$ for all $n \in \mathbb{N}$, form an irreducible sofic shift space. 
\end{proposition}

\begin{proof}
First, note that
$\sigma^p(x) =x $ implies $x$ can be written as a tiling $\cdots bbbbb \cdots$,
where $b$ is a word of length $p$. Pick any $x_i$ as the first symbol $b_0$ in
the word $b$. Then $\sigma(x)_i = b_1$, $\sigma^2(x)_i =b_2$, ..., with
$\sigma^{p-1} = b_p$. Applying $\sigma$ one more time gives $\sigma^p(x)_i =
x_i = b_0$, arriving at the next tile $b$.

Second, using this observation we create $\Presentation(\ShiftSpace)$ using
$p+1$ internal states, where we have a state $\xi_{b_i}$ for each symbol
$b_i \in \Alphabet$ in the tile $b$ (there are $p$ of these) and one absorbing
state $\xi_e$. Let $\xi_{b_0}$ be the future set equivalence class $[b]_F$ of the word
$b$. Due to $x$'s exact symmetry, there is one and only one symbol $a \in
\Alphabet$ such that $ua \neq e$ for $u \in [b]_F$; namely $b_1$. Therefore,
let $M_{b_1}(\xi_{b_0}) = \xi_{b_1}$, and then all other $M_a$ map to $\xi_e$,
for $a \neq b_1$. Now, $\xi_{b_1} = [bb_1]_F$ and $b_2$ is the only symbol
such that $ua \neq e$ for $u \in [bb_1]_F$. Thus, $M_{b_1}(\xi_{b_1}) =
\xi_{b_2}$ and all other $M_a$ map to $\xi_e$, for $a \neq b_2$. Repeat this
argument for all generators in $b$, with $\xi_{b_p} = [bb_1b_2 \dots b_p]_F =
[bb]_F$.

Third, as with $\xi_{b_0}$, the only symbol that can follow is $b_1$ and we
repeat the full argument again, where each $b_i$ sequentially follows.
Therefore, $[bb]_F$'s future set equals that of $[b]_F$ and so $[bb]_F=[b]_F$.
Thus, the only transition from $\xi_{b_p}$ that is allowed returns to the
original state $\xi_{b_0}$. This completely specifies
$\Presentation(\ShiftSpace)$ with a finite number of states $\Xi = \{\xi_{b_0},
\dots, \xi_{b_p}, \xi_e\}$. Given Theorem (1) above, one concludes that the shift space
$\ShiftSpace$ for a translation symmetric sequence is sofic. 
\end{proof}

Figure~\ref{fig:machines}(a) gives an example of $\Presentation(\ShiftSpace)$
for the set of translation symmetric strings $\cdots000111000111000\cdots$ with
$b=111000$. For visual clarity it omits the self-loops on state $\xi_e$.

Recall that $\Presentation(\ShiftSpace)$ is the irreducible component of the
future cover, so for a given shift there may be additional states beyond those
just described. However, these are equivalence classes for words of length less
than $p$, since we started with $[b]_F$. Since $p$ is finite, there are only
finitely many additional equivalence classes. These correspond to transient
(nonrecurrent) states of the future cover semiautomaton.

In Fig.~\ref{fig:machines}(a)'s example, knowing $x_i = 0$ does not fully
specify which state $\Presentation(\ShiftSpace)$ is in, since three states
$\xi_3$, $\xi_4$, and $\xi_5$ correspond to observing the symbol $0$. The
additional transient states of the future cover specify how to
\emph{synchronize} to $\Presentation(\ShiftSpace)$'s states from the generators
$a \in \Alphabet$ (single-symbol words).

Having constructed the canonical machine presentation
$\Presentation(\ShiftSpace)$, we can further relate $\ShiftSpace$'s semigroup
action to $x$'s translation symmetry group. For each state $\xi_{b_i} \in \Xi$
there is one and only one transition that does not lead to the absorbing
forbidden state $\xi_e$. That is, only one generator $a \in \Alphabet$ can be
concatenated to the words $u \in [u]_F = \xi_{b_i}$. Similarly, if we consider
a word $u \in \xi_{bi}$ as the window $u = \word{x}{i}{j}$, then a unit shift
by $\sigma$ reveals one and only one new symbol at index $j$ in
$\word{\sigma(x)}{i}{j}$. Therefore, ignoring the absorbing state and
transitions to it, the graph of $\Presentation(\ShiftSpace)$ is a cyclic graph
with period $p$: Every $p$-length path from $\xi_i$ returns to $\xi_i$ for all
$\xi_i \in \Xi \setminus \{\xi_e\}$. Thus, the permutation symmetries of this
(edge-labeled) graph correspond to elements of $x$'s translation symmetry
group.

In Fig.~\ref{fig:machines}(a)'s example, the state labeled $\xi_{b_0}$
corresponds to the start of the tile $b=111000$, but we could equivalently use
$\xi_{b_0}$ as the start of tile $000111$. Furthermore, the internal states
have a functional meaning. They are the elements of the quotient group of
the translation symmetry---counters that track the symmetry's phase.

It must be emphasized that sofic shifts and their semigroups do not formally
generalize such exact symmetries in the obvious way. That is, the semigroup of a
sofic shift for exact symmetry strings does not become a symmetry group. $G$'s
absorbing semigroup element $e$ is still required for an exact symmetry sofic
shift $\ShiftSpace$. From our construction of translation symmetric
$\Presentation(\ShiftSpace)$ we see that, for every internal state $\xi_{b_i}
\neq \xi_e$, there is one and only one transition that does not lead to $\xi_e$.
This makes it clear that exact symmetries are a highly restrictive form of
pattern. By representing exact symmetries using sofic shifts and their machine
presentations, though, it is now straightforward to generalize by relaxing the
restrictions that impose exact symmetries.

\subsection{Generalized Symmetries}

Ignoring $\xi_e$, a sofic shift $\ShiftSpace$ whose machine presentation is not
cyclic then represents a pattern as a generalized symmetry. We can again either
consider $\ShiftSpace$ as an ensemble of strings or one infinite sequence that
possesses a generalized symmetry described by $\ShiftSpace$'s semigroup, which
is well represented by the machine presentation $\Presentation(\ShiftSpace)$.
By removing the restriction of a cyclic graph in
$\Presentation(\ShiftSpace)$---that imposed the perfect regularity of $x =
\cdots bbbbb \cdots$---we now can capture a much wider class of patterned
strings with approximate or partial regularities. 

Consider the extreme case of the full shift $\FullShift$ that has no
regularity. There are no products $uv = e$ in $G$ for $\FullShift$, and so
there are no restrictions on its words. The algebra of the full shift is the
free semigroup. Its machine presentation $\Presentation(\FullShift)$ is a
single state with all $M_a$ mapping that one state back to itself---i.e., all
transitions are self transitions. Since all words can be concatenated to each
other, they all belong to a single future equivalence class.

We interpret $\FullShift$'s complete lack of regularity to be a null pattern.
Analogously, the opposite extreme of total regularity with strings of the form
$x = \cdots aaaaa \cdots$ for $a \in \Alphabet$ is also a null pattern with a
single-state (again, ignoring $\xi_e$) machine presentation that has a single
self-transition. The null pattern, in these cases, has zero memory---the
logarithm of $|\Xi \setminus \{\xi_e\}|$ vanishes. While both extremes at first
seem to be polar opposites, recall our goal is that ``pattern'' represents a
\emph{predictive} regularity and this is lacking in both cases. The full shift
is completely ``random'' and thus unpredictable. Whereas, for trivially
translation symmetric strings $x=\cdots aaaa \cdots$, the future is always the
same. There is nothing to predict.


Between the complete regularity of exact symmetries and lack of predictive
regularity in null patterns, we identify several categories of partial
predictive regularity.

First, note that for translation symmetric strings $\sigma^p(x)_i = x_i$ for all $i$.

Second, there are string classes for which $\sigma^p(x)_i = x_i$ for only
\emph{some} $i$. We call these \emph{partial symmetries}. A particular
case of partial symmetries are \emph{stochastic symmetries}. For simplicity,
consider binary sequences with $\Alphabet = \{0,1\}$, and let $\omega$ denote a
``wildcard'' that can be either $0$ or $1$ \cite{Marq11a}. Sofic shifts with
stochastic (partial) symmetries are fully translation symmetric after making
wildcard substitutions. For example, we can specify a sofic shift with
sequences of the form $x = \cdots \omega \omega \; 0 \; \omega \omega \; 0 \;
\omega \omega \; 0 \; \cdots$, say. Examples of such strings are $\cdots 11 \;
0 \; 01 \; 0 \; 00 \; 0 \; 10 \; 0 \; 11 \; 0 \cdots$, where spaces help
emphasize the ``fixed'' $0$s that are the scaffolding of the partial symmetry.
Note that the canonical machine presentation $\Presentation(\ShiftSpace)$ for
such stochastic symmetries are also cyclic graphs, as shown in
Fig.~\ref{fig:machines}(b).

Third, recall that if the canonical presentation $\Presentation(\ShiftSpace)$
is a cyclic graph, every $p$-length path from $\xi_i$ returns to $\xi_i$ for
all $\xi_i \in \Xi \setminus \{\xi_e\}$. Similar to how we generalized from
exact to partial symmetries, we define \emph{hidden symmetries} for which only
\emph{some} states $\xi_i \in \Xi \setminus \{\xi_e\}$ return to themselves on
all $p$-length paths in $\Presentation(\ShiftSpace)$. We exclude the case of
$p=1$, so that self-loops do not count as a hidden symmetry.
Figure~\ref{fig:machines}(c) shows an example with $\xi_2$ as the symmetric
state. The canonical machine presentation specifies a sofic shift consisting of
arbitrary arrangements of blocks $a=000$ and $b=111$; e.g., $x = \cdots aababba
\cdots = \cdots 000000111000111111000 \cdots$. The exact symmetry shift in
Fig.~\ref{fig:machines}(a) is the special case of the symmetric tiling $x =
\cdots abababa \cdots$.


Finally, Fig.~\ref{fig:machines} (d) gives an example of a general nonnull
pattern that is not an exact, partial, or hidden symmetry. This is the
well-known Even Shift \cite{Weis73,Kitc86a}---the set of binary strings in
which only even blocks of $1$s bounded by $0$s are allowed. This $\ldots 0
1^{2n} 0 \ldots$ pattern extends to arbitrary lengths, despite being specified
by only two internal states. While there are no states in the presentation
$\Presentation(\ShiftSpace)$ that always return to themselves after $p \neq 1$
transitions, there is still predictive regularity. In particular, if a $1$ is
seen after a $0$, it is guaranteed that the next symbol will be a $1$. This is
specified by $\xi_B$ having only one allowed transition to $\xi_A$ on a $1$.
Appendix \ref{app:EvenShift} discusses this example in more detail, along with
its semigroup and three semiautomaton presentations.

Before moving to probabilistic patterns represented by measure sofic shifts, we
note that Krohn-Rhodes theory \cite{Kroh65a,Rhod71a} was the first to connect
finite automata with a semigroup algebra. Moreover, it showed that finite
semigroups and their corresponding automata naturally decompose into simpler
components, including finite simple groups. This is yet another perspective
showing finite automata and their semigroup algebra capture patterns as
generalized symmetries. Further exploration of the connection between
Krohn-Rhodes theory and the perspective developed here is left for future work. 
One important difference to note is that their approach did not address
statistical or noisy patterns, as the following now does.

\subsection{Statistical Patterns Supported on Sofic Shifts}

The exposition on sofic shift patterns did not invoke probabilities over
symbols, words, or strings. Shift spaces are not concerned with the
probability of a word occurring, only whether a word can possibly occur or not.
This is why we referred to sofic shifts as \emph{topological} patterns. We just
saw that exact symmetries are given as sofic shifts and so are topological
patterns. Recall that our key motivation for showing this was to argue for sofic
shifts as a mathematical formalism that captures a notion of (topological)
pattern that greatly expands exact symmetry to generalized symmetry. However,
as we now describe, we can generalize further to formalize \emph{statistical}
patterns that are supported on sofic shifts. Doing so provides a direct link
with the more familiar statistical measures of order and organization used in
statistical mechanics, such as correlation functions. 

Our goal is to build a probability space on top of a shift space. The key
property of this probability space is that words are assigned positive
probability if and only if they are allowed in the shift space. This is
accomplished through the use of cylinder sets as the sigma algebra on a shift
space $\ShiftSpace$~\cite{Boyl09}. For a shift space $\ShiftSpace$ and a
length-$n$ word $\omega$, the \emph{cylinder set} $C_i(\omega)$ is defined as:
\begin{align*}
    C_i(\omega) := \{x \in \ShiftSpace \; : \; x[i: i+n-1] = \omega\}
    ~.
\end{align*}
Naturally then, a probability measure $\mu$ is assigned such that the
probability of a word $\omega$ occurring is given as:
\begin{align*}
    \Pr(\omega) = \mu\bigl(C_i(\omega)\bigr)
    ~.
\end{align*}
By definition, $\Pr(\omega) = 0$ if $\omega \in \Forbidden$, as the associated cylinder sets will be empty for forbidden words. 

Such a probability measure defines a \emph{stationary stochastic process} over the shift space $\ShiftSpace$ if it is shift-invariant, such that the probability of a word is independent of the index $i$ in $C_i(\omega)$, and each word satisfies prefix and suffix marginalization: 
\begin{align*}
\mu(w) & = \!\!\!\!\!\!\! \sum_{\{a: aw \in \Language(\ShiftSpace) \}}
	\!\!\!\!\!\! \mu(aw) \\
\end{align*}
and:
\begin{align*}
\mu(w) & = \!\!\!\!\!\!\! \sum_{\{a: wa \in \Language(\ShiftSpace) \}}
	\!\!\!\!\!\!\! \mu(wa)
  ~.
\end{align*}
This is the Kolmogorov extension theorem that guarantees the finite-dimensional
word distributions consistently define a stochastic process \cite{Dudl02a}. We
only consider shift-invariant measures and, so without loss of generality, we
simply use $C_0(\omega)$ for the cylinder sets. 

Crucially, the semigroup algebra and canonical machine presentation machinery for topological patterns have natural generalizations to stochastic patterns, as we now describe. 

Following Ref.~\cite{Furs60a} in the context of shift spaces, a (free)
\emph{stochastic semigroup} is a function $F$ defined on the free semigroup
$S$, with identity element $\eta$ (the empty symbol), that satisfies the
following properties (\cite[Definition 4.29]{Boyl09}):
\begin{enumerate}
      \setlength{\topsep}{-2pt}
      \setlength{\itemsep}{-2pt}
      \setlength{\parsep}{-2pt}
\item $F(\eta) = 1$,
\item $F(s) \geq 0$ for each $s \in S$ and $F(a) \geq 0$ for each
	$a \in \Alphabet$, and
\item $\sum_{a_i \in \Alphabet} F(a_i s) = \sum_{a_i \in \Alphabet} F(s a_i)$
	for each $s \in S$.
\end{enumerate}

For a sofic shift $\ShiftSpace$, we define a stochastic semigroup on
$\ShiftSpace$ using the semigroup $G$ defined above, with absorbing element $e$
corresponding to forbidden words. We simply set $F(e) = 0$. Then, a
shift-invariant measure $\mu$ satisfying Kolmogorov extension on a sofic shift
$\ShiftSpace$ with semigroup $G$ forms a stochastic semigroup as follows. For
all elements of the free semigroup $a_0a_1\cdots a_k$---i.e., for every word
$\omega = a_0a_1\cdots a_k$---define $F$ such that $F(\eta) = 1$, $F(\omega) =
0$ if and only if $\omega \in \Forbidden$ (equivalently, $a_0a_1\cdots a_k = e$
in the semigroup $G$), and otherwise $F(\omega) = \mu\bigl(C_0(\omega)\bigr)$.
Such a measure $\mu$ is called a \emph{sofic measure} \cite{Boyl09}. 

In this way, sofic measures allow for statistical structure on top of a sofic
shift $\ShiftSpace$, while maintaining an algebraic structure related to
$\ShiftSpace$'s semigroup algebra. More importantly, there is a canonical
machine presentation associated with sofic measures, analogous to the canonical
machine presentation of sofic shifts. As in the topological case for sofic
shifts, the canonical machine presentation of sofic measures provides the
mathematical formulation of \emph{statistical patterns}. 

Recall that the future cover semiautomaton---the canonical topological machine
presentation---is defined from Eq.~(\ref{eq:future_equiv})'s future equivalence
relation. The canonical stochastic machine presentation is defined through a
stochastic generalization of future equivalence, called \emph{predictive} or
\emph{causal} equivalence, defined on semi-infinite words. Each index $i$
partitions a sequence into a semi-infinite \emph{past} $\past_i = \{x_j\}, j
\leq i$, and semi-infinite \emph{future} $\future_i = \{x_k\}, k > i$. In the
topological case, two pasts are considered future-equivalent if they have the
same future---the same \emph{set} of futures that follow. In the stochastic
case, two pasts are \emph{predictively} or \emph{causally equivalent} if they
have the same \emph{distribution} $\Pr(\Future | \Past)$ over futures
conditioned on the past:
\begin{align}
\past_i  \! \sim_\epsilon  \! \past_j \iff
	\Pr(\Future | \Past \! = \! \past_i) = \Pr(\Future | \Past \! = \! \past_j)
    .
\label{eq:causal_equiv}
\end{align}
Just as the future equivalence relation $\sim_F$ defines the unique minimal
semiautomaton presenting a sofic shift, the causal equivalence relation
$\sim_\epsilon$ defines the unique minimal \emph{hidden Markov chain} (HMC)
that presents a sofic measure and its stationary stochastic process
\cite{Crut88a,Shal98a}.

Speaking simply, a hidden Markov chain $(\Xi, \Alphabet, \mathcal{T})$ is a
semiautomaton whose deterministic symbol-labeled transitions $\mathcal{M} =
\{M_0, M_1, \ldots, M_{(n-1)}\}$ are replaced by symbol-labeled transition
probabilities $\mathcal{T} = \{T^0, T^1, \ldots, T^{(n-1)}\}$, where $T^a_{\xi,
\xi'}$ is the probability of transitioning from state $\xi$ to $\xi'$ on the
symbol $a \in \Alphabet$. 

Paralleling the topological setting, the canonical stochastic machine
presentation is a hidden Markov chain whose internal states $\Xi$---the
\emph{predictive} or \emph{causal states}---are the equivalence classes of Eq.
(\ref{eq:causal_equiv})'s causal equivalence relation. The symbol-labeled
transitions are then defined through the one-step conditional distributions
$\Pr(X^1 = a | \Past=\past_i)$. For a given causal state $\xi$, we write
$\Pr(\Future | \Past = \xi)$ since by definition each past $\past_i$ in the
equivalence class $\xi = [\past_i]_\epsilon$ has the same predictive
distribution.

The transition probability $T^a_{\xi, \xi'}$ is then given as $\Pr(X^1 = a |
\Past = \xi)$, with $\xi' = [\past_i a]_\epsilon$, where $\past_i a$ is the new
past given by concatenating the observed symbol $a$ onto the current past
$\past_i$. This follows from \emph{unifilarity}~\cite{Shal98a}: in the
stochastic setting for each internal state $\xi \in \Xi$ and symbol $a \in
\Alphabet$ there is at most one internal state $\xi'$ such that $T^a_{\xi,
\xi'} > 0$. It then follows that $[\past_i a]_\epsilon$ is the same for all
$\past_i \in \xi$.


Historically, the canonical stochastic machine presentation $(\Xi, \Alphabet,
\mathcal{T})$ is known as the \emph{$\epsilon$-machine} \cite{Shal98a} of the
associated stochastic process. As described shortly, the $\epsilon$-machine is
a generator of its associated statistical field theory. It generates all words
with their corresponding probabilities. Similar to its topological counterpart,
unifilarity additionally elevates the \eM\ to a predictive presentation.
Prediction in the stochastic setting means identifying the predictive
distribution associated with a given past. And, by definition, an \eM's causal
states carry unique predictive distributions.

Summing over all symbol-labeled transitions produces a Markov transition
operator over the internal causal states: $T_\epsilon = \sum_{a \in \Alphabet}
T^a$. This operator evolves probability distributions over the causal states,
regardless of the symbols involved, and specifies an order-1 Markov process
over the states~\cite{Shal98a}. The left eigenvector of $T_\epsilon$ with unit
eigenvalue provides the \emph{stationary distribution} over causal states: $\pi
T_\epsilon = \pi$. For ergodic systems, as we consider here, $\pi$ is unique. 

We emphasize that the causal states identify a hidden, internal Markov process
underlying the non-Markovian process over the symbols in $\Alphabet$,
specified by a sofic measure on a sofic shift. This inverts the abstract
definition of sofic measures, given as factor maps on Markov
processes~\cite{Boyl09}. There, a shift-invariant measure obeying Kolmogorov
extension on a subshift of finite type produces a Markov chain of finite order
and sofic shifts are abstractly defined as factor maps of subshifts of finite
type. Hidden Markov chains are then given as the pushforward of the Markov
measure along the factor map. Here, in contrast, we start with a statistical
field theory supported on a sofic shift. The causal equivalence relation then
identifies the underlying causal states and the Markov process defined over
them.

With our given assumption of a stationary ergodic process over symbols we can
use the stationary distribution $\pi$ and the symbol-labeled transition
operators to directly extract the word probabilities $\Pr(\omega) =
\mu\bigl(C_0(\omega)\bigr)$ from the $\epsilon$-machine presentation. First,
single-symbol probabilities are given as:
\begin{align}
    \Pr(a) = \sum\limits_{\xi \in \Xi} \pi(\xi) \sum\limits_{\xi' \in \Xi}T^{a}_{\xi \xi'}
    ~. 
\end{align}
Where $\sum_{\xi'}T^a_{\xi \xi'}$ gives the probability of observing the symbol
$a \in \Alphabet$, conditioned on being in causal state $\xi$, and this is
summed over all states in $\Xi$ weighted by their stationary probabilities
$\pi$. We write this compactly as:
\begin{align}
    \Pr(a) = \langle \pi | T^a | \mathbf{1} \rangle
    ~,
\end{align}
where $\langle \pi |$ indicates $\pi$ as a row vector and $| \mathbf{1} \rangle$ is the column vector of all $1$s. 

The probability of a word $\omega = a_0a_1\ldots a_k$ is then:
\begin{align}
\Pr(a_0a_1\cdots a_k) &= \mu \bigl(C_0(a_0a_1\ldots a_k)\bigr) \nonumber \\ 
&= \langle \pi | T^{a_0} T^{a_1} \cdots T^{a_k} | \mathbf{1} \rangle
~.
\label{eq:word_prob_machine}
\end{align}
Recall that in the topological case, the semigroup algebra of the canonical
machine presentation is given through composition of the mappings $\mathcal{M}
= \{M_a\}$. Now, we see the same semigroup algebraic structure in the products
of the symbol-labeled transition matrices $\mathcal{T} = \{T^a\}$. In fact, the
topological structure, the set of mappings $\mathcal{M}$, is recovered from the
statistical structure, the set of transitions $\mathcal{T}$, by setting all
nonzero elements of each $T^a \in \mathcal{T}$ to unity and then applying
future equivalence. 

It must be emphasized that this last step, applying future equivalence, is
essential. There may be distinct $\epsilon$-machines---presentations of
statistical patterns---that are supported on the same sofic shift---topological
pattern. Said another way, statistical patterns signify distinct structure
supported on topological patterns. They are not merely ``adding probabilities
onto'' topological patterns. Formally, the symbol-labeled transition matrices
$\mathcal{T}$ can represent a different semigroup than that represented by
their topological counterparts $\mathcal{M}$ on which they are supported.
Appendix~\ref{app:statpatterns} illustrates this distinction.

Let us briefly turn to mention the quantitative benefits of having these
presentations. Using the word probabilities as just described, in addition to
the underlying pattern of the sofic shift that supports the ensemble, a wide
range of statistical properties of the sequence ensembles, such as correlation
functions, power spectral densities, and informational properties
\cite{Riec13a,Riec17a} can be directly determined via the \eM. For example, the
\emph{Shannon entropy rate} $\hmu$ of an ensemble measures its degree of
randomness and can be calculated as:
\begin{align*}
\hmu & = \lim_{\ell\rightarrow\infty} \frac{\H[\word{X}{0}{\ell}]}{\ell} \\
	& = \lim_{\ell\rightarrow\infty} \H[X_{0} | \word{X}{-\ell}{0} ]  \\
    & = - \sum_{\xi \in \Xi} \pi(\xi) \sum_{a \in \Alphabet}
	\sum_{\xi' \in \Xi} T^{(a)}_{\xi \xi'} \log_2 T^{(a)}_{\xi\xi'}
~,
\end{align*}
where $X_{0:\ell}$ for $\ell > 0$ is the random variable for the subsequence of
symbols $x_0 x_{1} \cdots x_{\ell-1}$. Moreover, the Shannon information in the
causal states measures a process' historical memory. For more nuanced
information measures of patterns in stochastic processes and spin systems, see
Refs. \cite{Crut01a,Feld02b,Feld08a,Crut08a,Crut10a,Jame11a,Crut13a,Ara14a}.

\section{Patterns in Spacetime}
The preceding demonstrated that sofic shifts and their semiautomaton
presentations provide a formulation of patterns in discrete one-dimensional
systems. And, it showed how they generalize to stochastic ensembles with
associated information-theoretic measures. With this, one can argue that the
theory of one-dimensional patterns in discrete stationary processes, augmented
with the cited extensions, is largely complete. Now, we turn to the question of
how to capture patterns in higher dimensions. While patterns in this new
setting are amenable to similar analysis, there are key differences, new
phenomena, and open problems.


Consider now the time evolution of symbols $\Alphabet$ on the sites of a
spatial lattice $\lattice$. A \emph{spacetime field} $\stfield{}{} \in
\MeasAlphabet^{\lattice \otimes \mathbb{Z}}$ is a time series $\state_0,
\state_1, \ldots$ of spatial configurations $\state_t \in
\MeasAlphabet^\lattice$. With a $d$-dimensional lattice $\lattice =
\mathbb{Z}^d$ a spacetime field $\stfield{}{}$ is an element of a
\emph{spacetime shift space} $\FullSTShift$ \cite{Schmi95a,Lind04a}. In what
follows, upper indices on field values are spatial coordinates (e.g., at time
$t$ site $i$ has value $\state^i$) and lower indices are time coordinates
(e.g., $\state_t$).

Multidimensional shift spaces are notoriously difficult to study, with many
simple properties being uncomputable \cite{Lind04a}. Similarly, while sofic
shifts and their canonical machine presentations provide a mathematical
formalism for patterns as generalized symmetries in one dimension, generalizing
to higher dimensions is not straightforward. Significantly, there is not a
unique generalization for finite-state machines and regular languages in higher
dimensions \cite{Lind98a}. If there were a unique generalization, we could use
the semiautomata special case as the mathematical representation of
high-dimensional patterns. 

As we argued, models that create a compressed representation of a system's
behavior must do so by harnessing patterns---predictive regularities---in the
system. Fully-discrete one-dimensional systems are ideal as there is a unique
minimal presentation of the system and, in this case, \emph{that
predictive presentation is the pattern}. The situation is more complicated in
higher dimensions. A conflict arises between useful generalized spacetime
symmetries and predictive presentations that faithfully generate fields in their
spacetime shift spaces.

We now outline the local spacetime generalization of predictive equivalence for
constructing presentations of spacetime patterns. As with finite-state machines
and regular languages, degeneracy is broken when moving to higher dimensions as
the spacetime generalization is not unique. In particular, we demonstrate that
the shape of local ``futures'' determines the algebraic properties of the
resulting local presentation.

\subsection{Local Spacetime Presentations}

Since an evolving spacetime field is a time series of spatial lattice
configurations, it can be interpreted as a one-dimensional shift space over an
exponentially-large alphabet (of lattice configurations). While formally well
defined, however, this perspective is an unwieldy basis for a mathematical
formulation of patterns in spacetime. First, for space and time translation
symmetries, we typically consider the idealized case of infinitely-large
spatial lattices. Thus, one must work with shift spaces over an
infinitely-large alphabet. Second, capturing patterns \emph{within} the spatial
configurations themselves means not treating an entire lattice configuration
simply as a single symbol, as done with the one-dimensional framing. There is
internal organization within the spatial configurations that is dynamically
and structurally important.

Therefore, we take a local approach to generalizing machine presentations of
spacetime shift spaces \cite{Shal03a, Rupe18a}. This is implemented by
identifying equivalence classes over \emph{local} pasts that have the same
future set of \emph{local} futures in spacetime. The latter can be either
topological (sets) or statistical (predictive distributions).

Motivated by the causal restrictions of local interactions, \emph{past
lightcones} in spacetime are the natural choice for local pasts when building
local presentations. Formally, the past lightcone $\plc$ of a site $\point$ in
spacetime is the set of all field values $\stpointprime$ at previous times 
$(t^\prime \leq t)$ that could influence $\stpoint$ through local interactions:
\begin{align}
\plc \point \equiv \big\{ \stpointprime :\; t' \leq t \; \mathrm{and} \;
|\site' - \site| \leq c(t - t') \big \}
  ~,
\label{eq:plc}
\end{align}
where $c$ is the finite speed of information propagation in the system.  We
include the \emph{present} field value $\stpoint$ in $\point$'s past lightcone,
but not in its local future.

Due to the richness and complication of multidimensional shift spaces, the
following explores only topological spacetime patterns. Paralleling the
one-dimensional development, local presentations of topological spacetime
patterns are defined through the local analog of the future equivalence
relations:
\begin{align}
\plc_i \sim_F \plc_j \iff F(\plc_i) = F(\plc_j)
  ~,
\label{eqn:localequiv}
\end{align}
where $F(\plc_i) = \{$Local futures co-occurring with $\plc_i\}$. We define
co-occurring local futures more precisely shortly, as there are alternatives.

Relation (\ref{eqn:localequiv})'s equivalence classes determine the internal
states for a local spacetime presentation. The generalization of symbol-labeled
transitions between the internal states is constructed in terms of spatial
$[\sigma^s(\stfield{}{})]_t^r = \stfield{t}{r+s}$ and temporal
$[\sigma_\tau(\stfield{}{})]_t^r=\stfield{t+\tau}{r}$ shift operators. (Note
that all space-time shift operators commute.) For a past lightcone $\plc\point$ at
spacetime site $\point$, let $\sigma^s_\tau\bigl(\plc\point\bigr)$ denote the
action of the shift operator on all of the field values in
$\plc\point$. That is, an entire lightcone is shifted analogously to an
individual spacetime site.

\begin{figure*}
\centering
\includegraphics[width = 0.9 \textwidth]{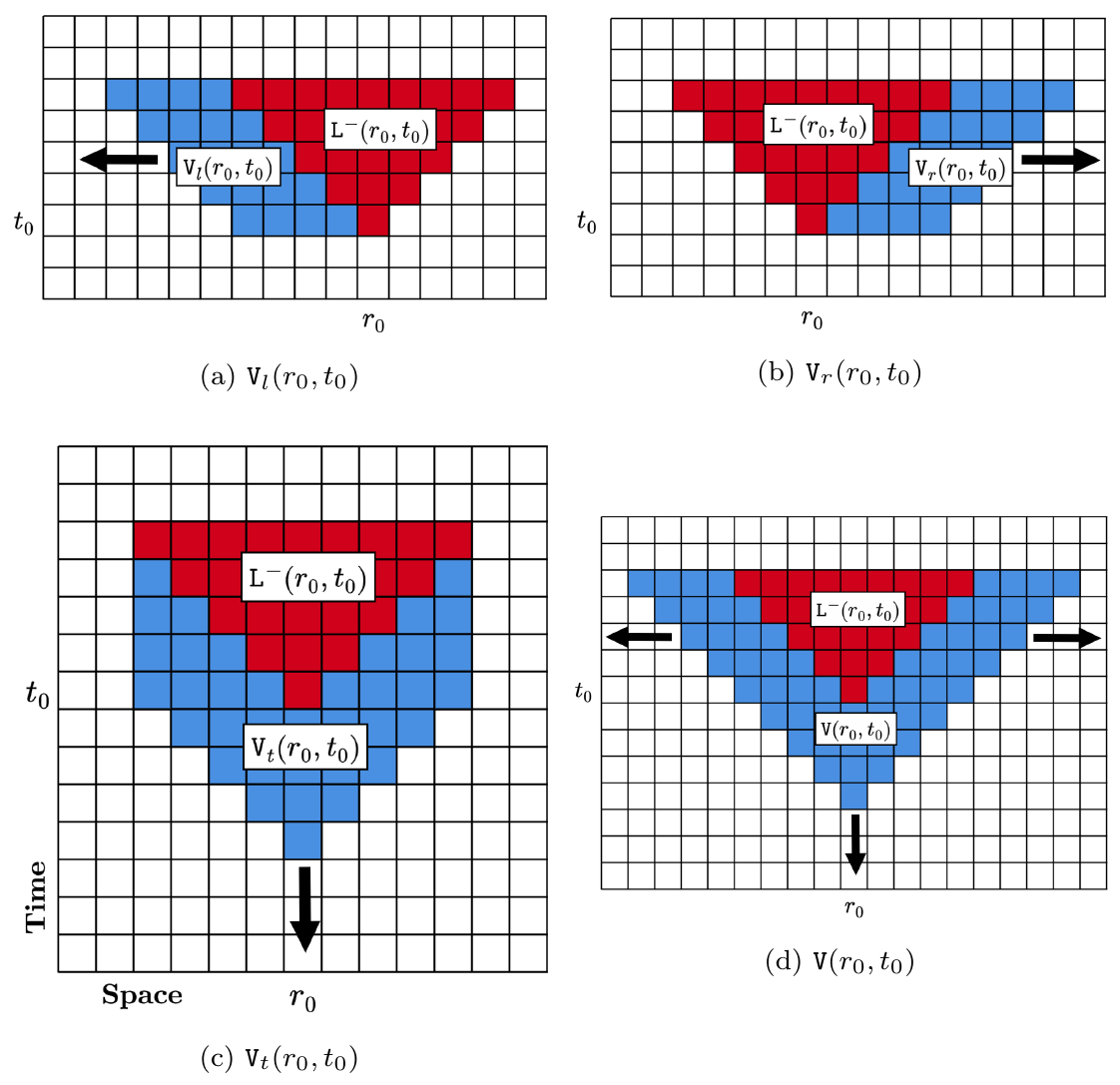}
\caption{Fringes induced by spacetime shifts: Co-occurring depth-4 past
	lightcone $\plc(r_0,t_0)$ (red) and depth-$4$ spacetime patches resulting
	from concatenations of (a) left transition fringes $V_\ell(r_0,t_0)$
	(blue), (b) right transition fringes $V_r(r_0,t_0)$ (blue), (c) forward
	transition fringes $V_t(r_0,t_0)$ (blue), and (d) unions of left, right,
	and forward transition fringes $V(r_0,t_0)$ (blue). Arrows indicate the
	direction(s) in which local spacetime patches may be generated with
	successive concatenations to the seed past lightcone $\plc(r_0,t_0)$. 
	}
\label{fig:Vs}
\end{figure*}

The right spatial transition \emph{fringe} is defined as the set difference
between $\plc\point$ and $\sigma^1\bigl(\plc\point\bigr)$. This is the
generalization of the ``symbol'' emitted during a rightwards move in one
spatial dimension. Similarly, the left spatial transition fringe is defined as
the set difference between $\plc\point$ and
$\sigma^{-1}\bigl(\plc\point\bigr)$, and the forward temporal transition fringe
is the set difference between $\plc\point$ and $\sigma_1\bigl(\plc\point\bigr)$
\cite[Fig. 4]{Shal03a}. For simplicity we consider spacetime fields with one
spatial dimension, but the generalization to higher dimensions is
straightforward. 

Time and space transition fringes form the appropriate alphabet to define local
spacetime presentations $(\Xi, \Alphabet, \mathcal{M})$, with $\Xi$ the set of
past lightcone equivalence classes, $\Alphabet$ the set of space and time
transition fringe, and $\mathcal{M}$ the fringe-labeled state transitions
\cite{Shal03a}. Briefly, a site $\point$ in spacetime has an associated
internal state $\xi\point$ that is the equivalence class of the past lightcone
$\plc\point$. Consider the neighboring site $(r+1, t)$, which similarly
has an associated internal state $\xi(r+1, t) = [\plc(r+1,t)]_F$. The right
transition fringe provides the missing information to construct $\plc(r+1,t)$
from $\plc\point$. Therefore, $\xi\point$ plus a right transition fringe
uniquely determines $\xi(r+1,t)$. That is, such local presentations are
\emph{unifilar}, and therefore have the requisite structure of predictive
presentations. Recall that predictive models are also generative. 

Generating words is rather straightforward for machine presentations in one
dimension, markedly less so for local spacetime presentations. The
fringe-labeled transitions establish local spacetime presentations $(\Xi,
\Alphabet, \mathcal{M})$ as local generative spacetime models. Local spacetime
patches, local ``words'', are generated via concatenation of transition fringe
``symbols''. Denote concatenations of right spatial fringes generally as
$\V_r$, to signify the particular shape of the resulting spacetime patches.
Similarly, let $\V_l$ be the spacetime patches resulting from concatenations of
left spatial fringes and $\V_t$ the patches from forward temporal fringes. 

Note that spacetime patches are more than merely collections of values from
spacetime sites. The configuration and space-time relation among the values
matter. This is why we refer to their \emph{shape}. The $\V_r$ patches are a
distinct form of local spacetime words apart from $\V_l$ and $\V_t$. Only
patches of the same shape may be concatenated together to extend patches.

Let $\V$ be the spacetime patches resulting from the union of forward, left,
and right transition fringes. Let \emph{depth} denote the number of
single-fringe symbols concatenated together in a particular spacetime patch
word. See Fig.~\ref{fig:Vs}.

Concatenating fringes into local spacetime patches has the same semigroup
algebra structure as in the one-dimensional setting. Similarly, the semigroup
algebra is captured by the local presentations through the action of the
fringe-labeled transitions $\mathcal{M}$, analogous to the one-dimensional
presentations $\Presentation(\ShiftSpace)$. We now examine possible choices of
local futures and the algebraic properties of the induced local presentations. 

\subsection{The Shape of Local Futures}

Recall that Relation (\ref{eqn:localequiv}) did not specify local futures when
defining local presentations. The future sets $F(\plc)$ of past lightcones
there were intentionally left ambiguous to allow for possible
alternatives---alternatives that arise to address several subtleties of
spacetime shifts.

Prior work assumed that the natural choice for a local future is a \emph{future
lightcone} $\flc$ \cite{Shal03a,Rupe18a,Rupe18b}. The latter is defined as all
field values at subsequent times that could possibly be influenced from the
given spacetime site $\stpoint$ through the local interactions:
\begin{align}
\flc \point \equiv \big \{ \stpointprime :\; t' > t \; \mathrm{and} \;
|\site' - \site| \leq c ( t' - t) \big \}
  ~.
\label{eq:flc}
\end{align}
Co-occurring past $(\plc)$ and future $(\flc)$ lightcones at spacetime point
$(\site_0, t_0)$ are depicted in Fig.~\ref{fig:lightcones} for a $1+1$
dimensional spacetime field with $c=1$. Local presentations constructed as
equivalence classes of past lightcones that have the same future set of future
lightcones are known as \emph{local causal states}.

\begin{figure}
\centering
\includegraphics[width = 0.48 \textwidth]{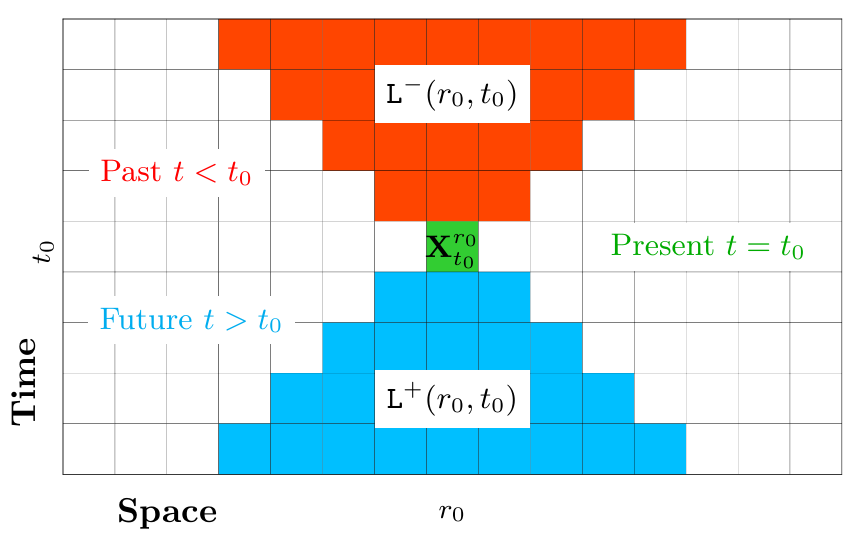}
\caption{Co-occurring past ($\plc$) and future ($\flc$) lightcones at a
	spacetime site $(r_0, t_0)$ in 1+1 dimensions with $c=1$.
	}
\label{fig:lightcones}
\end{figure}

Interestingly, the benefit of local causal-state models derives from
generalized symmetries \cite{Rupe18a,Rupe18b, Rupe19a}. However, this is a
spacetime symmetry distinct from the semigroup algebra of fringe
concatenations. While we describe these spacetime symmetries in more detail
below, we first demonstrate that such symmetries are at odds with the semigroup
algebra of fringe concatenations for local causal states. That is, while
employing lightcones as local futures allows one to discover useful spacetime
symmetries, the resulting local machine presentation $(\Xi, \Alphabet,
\mathcal{M})$ is \emph{not} a faithful generator of the underlying spacetime
shift space. Even though all local presentations possess the unifilarity
property of predictive models, the more basic generative ability falls short
unless the appropriate local futures are used. We again emphasize that
phenomena in higher dimensions are complicated---sometimes in counterintuitive
ways, as the following demonstrates.

Elementary cellular automata (ECA)---see App.~\ref{app:CA}---produce spacetime
shift spaces with nontrivial patterns and structure in their spacetime fields;
e.g., domains, particles, and particle interactions
\cite{Hans90a,Crut93a,Hans95a}. In addition, their one-dimensional
fully-discrete spatial lattices are shift spaces and so provide a link between
one-dimensional shifts and $1+1$ dimensional shifts of their spacetime fields. 

Following convention, denote the mapping from a past lightcone to its
equivalence class or, equivalently, to its associated local (causal) state
$\xi$, as $\epsilon(\plc) = [\plc]_F = \xi$. Crucially, this
provides a local pointwise mapping over any spacetime field $\stfield{}{}$.
Each site $\point$ has a unique past lightcone $\plc$ that is then mapped to
its local state via $\epsilon(\plc)$. We use this pointwise mapping
to transform a spacetime field $\stfield{}{}$ to an associated local-state
field $\causalfield{}{} = \epsilon(\stfield{}{})$ that shares the same
coordinate geometry as $\stfield{}{}$. Each site $\causalfield{r}{t}$ in the
local state field $\causalfield{}{}$ is the local state $\causalfield{r}{t} =
\xi = \epsilon(\plc)$ that is the image under the $\epsilon$-map of the
past lightcone $\plc$ at the spacetime point $\stpoint$. 

An example spacetime field $\stfield{}{}$ from ECA Rule 90 is shown in
Fig.~\ref{fig:rule90} consisting of black ($x_t^r = 1$) and white ($x_t^r = 0$)
squares. The associated local causal-state field $\causalfield{}{}$ is shown as
the overlaid colored letters, using lightcones as local futures. In the case of
Rule 90, there is a single local causal state, labeled as $\mathsf{A}$---all past lightcones are
future-lightcone equivalent.

This is consistent with prior findings \cite{Rupe18a,Rupe18b} that connect the
local causal states with the canonical machine presentations of one-dimensional
sofic shifts that represent \emph{invariant sets} of spatial configurations for
the ECA. The full $2$-shift is invariant under Rule 90, meaning there are no
forbidden words in the one-dimensional spatial configurations generated by Rule
90. All binary strings have pre-images under the global dynamic $\Phi$ of Rule
90. Recall that the canonical machine presentation of full shifts consists of a single internal state and, hence, we expect the single local causal-state for
the spacetime fields produced by Rule 90. 

\begin{figure*}
\centering
\includegraphics[width = 0.85 \textwidth, trim={5.0cm 2cm 5.0cm 2cm},clip]{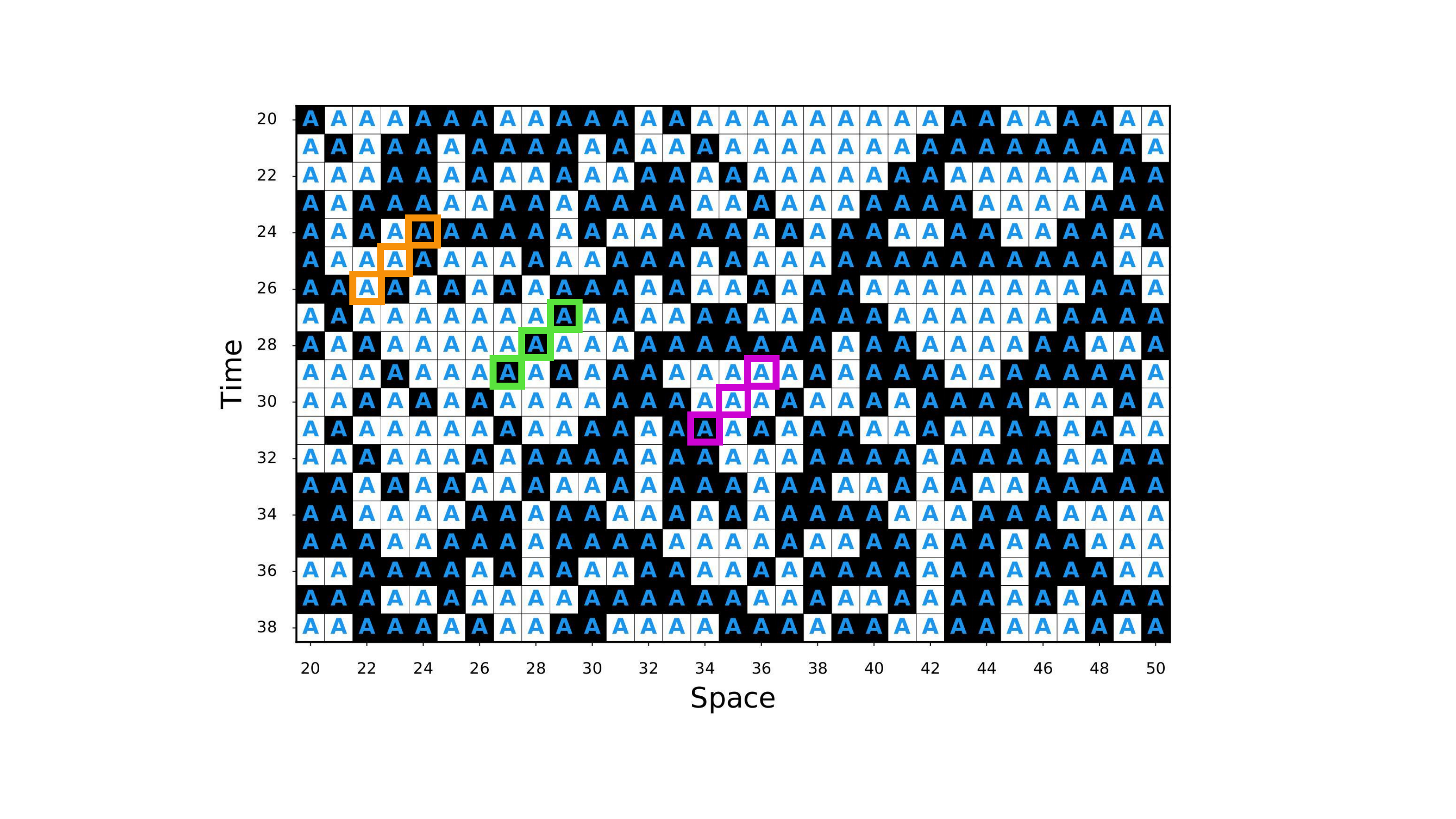}
\caption{ECA Rule 90 spacetime field depicted as white ($0$) and black ($1$)
	squares. The corresponding local causal-state field is overlaid with
	colored letters; simply the single causal state $A$. Three sample
	right-transition fringes for past lightcone depth-$2$ are highlighted in
	colored (orange, green, purple) boxes.
	}
\label{fig:rule90}
\end{figure*}

Even though there are no forbidden words in Rule 90's one-dimensional spatial
configurations, there certainly \emph{are} forbidden patches in
its spacetime fields. In particular, note that the ECA's local update rule 
$\phi$ corresponds to a spacetime patch in the shape of a depth-$1$ past
lightcone---the present spacetime site together with its local neighborhood one
time-step in the past. Therefore, any spacetime patch of the same shape that
is inconsistent with local update rule $\phi$ is forbidden in the spacetime
shift space of the ECA's spacetime fields. For example, the following patch is
allowed by Rule 90, since $\phi_{90}(010) = 0$:
\begin{align*}
\begin{matrix}
     0 & 1 & 0 \\
       & 0 &   
\end{matrix}
\end{align*}
Therefore: 
\begin{align*}
\begin{matrix}
     0 & 1 & 0 \\
       & 1 &   
\end{matrix}
\end{align*}
is a forbidden spacetime patch in Rule 90's spacetime shift space. 

From ECA Rule 90's single local causal-state it is easy to see that the local
causal-state machine presentation is not a consistent spacetime generator.
Since there is only one state, $\mathsf{A}$, all fringe-labeled
transitions occur from state $\mathsf{A}$ back to itself. Thus, all like-shaped fringe
symbols that occur may be concatenated together to form spacetime patches.
Figure~\ref{fig:rule90} highlights three right-transition fringes. If they are
concatenated in order of orange, green, magenta (left to right) they form the
depth-$3$ $\V_r$ patch:
\begin{align*}
\begin{matrix}
      &  & 1 & 1 & 0 \\
      & 0 & 1 & 0 &  \\
     0 & 1 & 1 &  &  
\end{matrix}
\end{align*}
This patch contains the spacetime word given above that is forbidden in Rule 90's spacetime shift. In fact, the other word of the same shape in this patch is also forbidden by Rule 90:
\begin{center}
    \begin{tabular}{c c c}
     1 & 1 & 0 \\
       & 0 &   
    \end{tabular}
    ~,
\end{center}
since $\phi_{90}(110) = 1$. 

Concatenations of left and forward fringes produce, respectively, $\V_l$ and
$V_t$ spacetime patches that similarly contain words forbidden by Rule 90's
spacetime shift space, as shown numerically in a supplementary Jupyter Notebook
\cite{rule90nb}. Similar results for another ECA shift space, the domain of
Rule 18 discussed in more detail below, are also given in a supplementary
Jupyter Notebook \cite{dom18nb}. Empirically, we have found that local causal
state presentations are generically not faithful generative models of their CA
spacetime shift spaces. 

Comparing Fig.~\ref{fig:lightcones} with the $\V_i$ patch shapes in
Fig.~\ref{fig:Vs}, there is little to no overlap between the ``futures''
(lightcones in this case) used to define the equivalence classes of past
lightcones and the spacetime patches generated from fringe-labeled transitions
between the equivalence classes. Therefore, it is not surprising that local
presentations created with future lightcones are not faithful generators of
their spacetime shifts. 

In one dimension, there is only one ``future'' that always follows from the
past. When translating forward, what once was a future becomes part of the
past. This is not necessarily the case for local pasts and futures in
spacetime. The predictive ability of unifilar models follows from this
succession of futures becoming pasts, when combined with future equivalence. As
we now see, future equivalence without the successional relation between local
pasts and futures yields models that, while unifilar, are not consistent
spacetime generators. 

Constructively, this insight points the way to creating local presentations
that \emph{are} faithful generators. As the future lightcones do not play a
role in generating local spacetime patches, they should not be used to define
past lightcone equivalence in Relation (\ref{eqn:localequiv}). Rather, we
should use spacetime patch shapes---denote them $\V_i$---that are to be
generated as our notion of a local future for defining future equivalence in
Relation (\ref{eqn:localequiv}). Local causal states are defined using the
future sets with future lightcones $F(\plc_i) = \{\flc \text{ co-occurring with
} \plc_i\}$. We can alternatively define local machine presentations whose
internal states are defined as equivalence classes from Relation
(\ref{eqn:localequiv}) using $F(\plc_i) = \{\V_i \text{ co-occurring with }
\plc_i\}$.

Indeed, as demonstrated computationally in the supplementary Jupyter Notebook
Ref. \cite{rule90nb}, using $\V_r$ shapes with $F(\plc_i) = \{\V_r \text{
co-occurring with } \plc_i\}$ yields a local machine presentation that is a
faithful generator of $\V_r$ spacetime patches. However, this presentation is
\emph{not} a faithful generator of $\V_l$ or $\V_t$ spacetime patches. Similar
to the use of future lightcones, local presentations can only generate faithful
spacetime patches in the shape used to define $F(\plc)$. Therefore, as also
demonstrated in Ref. \cite{rule90nb}, using $\V_l$ to define $F(\plc)$ results
in local presentations that faithfully generate $\V_l$, but not $\V_r$ nor
$\V_t$. Similar for $\V_t$. 

This motivates the definition of $\V$ in Fig.~\ref{fig:Vs} (d) as the union of
$\V_r$, $\V_l$, and $\V_t$. Local presentations defined using $F(\plc_i) = \{\V
\text{ co-occurring with } \plc_i\}$ are faithful generators of spacetime
patches in \emph{all} directions emanating from a seed past lightcone. As the
resulting presentations are also unifilar, additionally they are consistent
predictive models of local spacetime patches. These $\V$ presentations can thus
be seen as the most natural generalization of the predictive canonical machine
presentations $\Presentation(\ShiftSpace)$ for one-dimensional sofic shifts.
Moreover, they possess the same semigroup algebra of ``symbol'' (fringe)
concatenation that allows the generation and prediction of arbitrarily-long
``words''---arbitrarily large spacetime patches. 

That said, as the following now details, local causal-state presentations that
use future lightcones to define $F(\plc)$ possess interesting and useful
generalized spacetime symmetries that are lost when using $\V$ presentations. 

\subsection{Generalized Spacetime Symmetries} 

Recall from above that machine transitions in one-dimension determine the
structural relations among internal states and that they are driven by the
shift operator. We just saw that the analogous fringe-labeled transition
structure in spacetime is similarly related to spacetime shift operators
through the definition of the fringes. However, the spacetime shift operators
provide structural algebraic relations among local states \emph{independent of}
fringe symbols and their concatenated spacetime words. This is due to the local
pointwise $\epsilon$-map and the resulting shared coordinate geometry between
spacetime fields $\stfield{}{}$ and corresponding local state fields
$\causalfield{}{} = \epsilon(\stfield{}{})$.

To see this, consider two space-adjacent sites in a spacetime field,
$\stfield{t}{r}$ and $\stfield{t}{r+1} = \sigma^1 \stfield{t}{r}$. Under the
$\epsilon$-map, this produces $\xi_i = \causalfield{t}{r}$ and $\xi_j =
\causalfield{t}{r+1} = \sigma^1 \xi_i$. As past lightcones are defined solely
in terms of spacetime distances, they are equivariant under spacetime
isometries. In the present setting these are translations. 


For example, if a spacetime field $\stfield{}{}$ has \emph{exact} time and
space translation symmetries $\sigma_\tau\stfield{}{} = \stfield{}{}$ and
$\sigma^s\stfield{}{} = \stfield{}{}$, for some $s$ and $\tau$, then the
corresponding local state field $\causalfield{}{} = \epsilon(\stfield{}{})$
shares these symmetries: $\sigma_\tau\causalfield{}{} = \causalfield{}{}$ and
$\sigma^s\causalfield{}{} = \causalfield{}{}$. This is because the spacetime
shift operators act equivalently on lightcones, so that $\sigma^s\plc_i =
\plc_i$ and $\sigma_\tau\plc_i = \plc_i$. Therefore,
$\epsilon\bigl(\sigma^s\plc_i\bigr) = \epsilon(\plc_i)$ and
$\epsilon\bigl(\sigma_\tau\plc_i\bigr) = \epsilon(\plc_i)$. Note that this
argument is independent of which local future shape is used, as they all define
a local $\epsilon$-map on past lightcones. 

When future lightcones specifically are used as local futures, we previously
showed (i) there are spacetime fields $\stfield{}{}$ that do not have
translation symmetries, but (ii) the associated local causal state field
$\causalfield{}{} = \epsilon(\stfield{}{})$ \emph{does} \cite{Rupe18a,Rupe18b}.
These are the spacetime generalizations of partial and hidden symmetries; cf.
Fig. \ref{fig:machines}.

For the cellular automata examples in Refs. \cite{Rupe18a,Rupe18b} and shown
here in Fig.~\ref{fig:lcsfields}, generalized symmetries---exact, partial, and
hidden---in the spacetime fields are generated by the evolution of invariant
one-dimensional sofic shifts (spatial configurations) under the CA dynamic. Such
spacetime regions are known as \emph{domains} \cite{Hans90a,Crut93a,Hans95a}.
Interestingly, exact symmetry domains are generated from the evolution of exact
symmetry sofic shifts. Moreover, stochastic partial symmetry domains are
generated from stochastic partial symmetry sofic shifts and hidden symmetry
domains from hidden symmetry shifts. Examples of each of these cases are shown
in Fig.~\ref{fig:lcsfields} as the generalizations of the one-dimensional cases
in Fig~\ref{fig:machines}. Appendix~\ref{app:domain_machines} displays the presentations $\Presentation(\ShiftSpace)$ for each sofic shift $\ShiftSpace$ used to generate the fields in Fig.~\ref{fig:lcsfields}. 

\begin{figure*}
\centering
\includegraphics[width = 0.8 \textwidth]{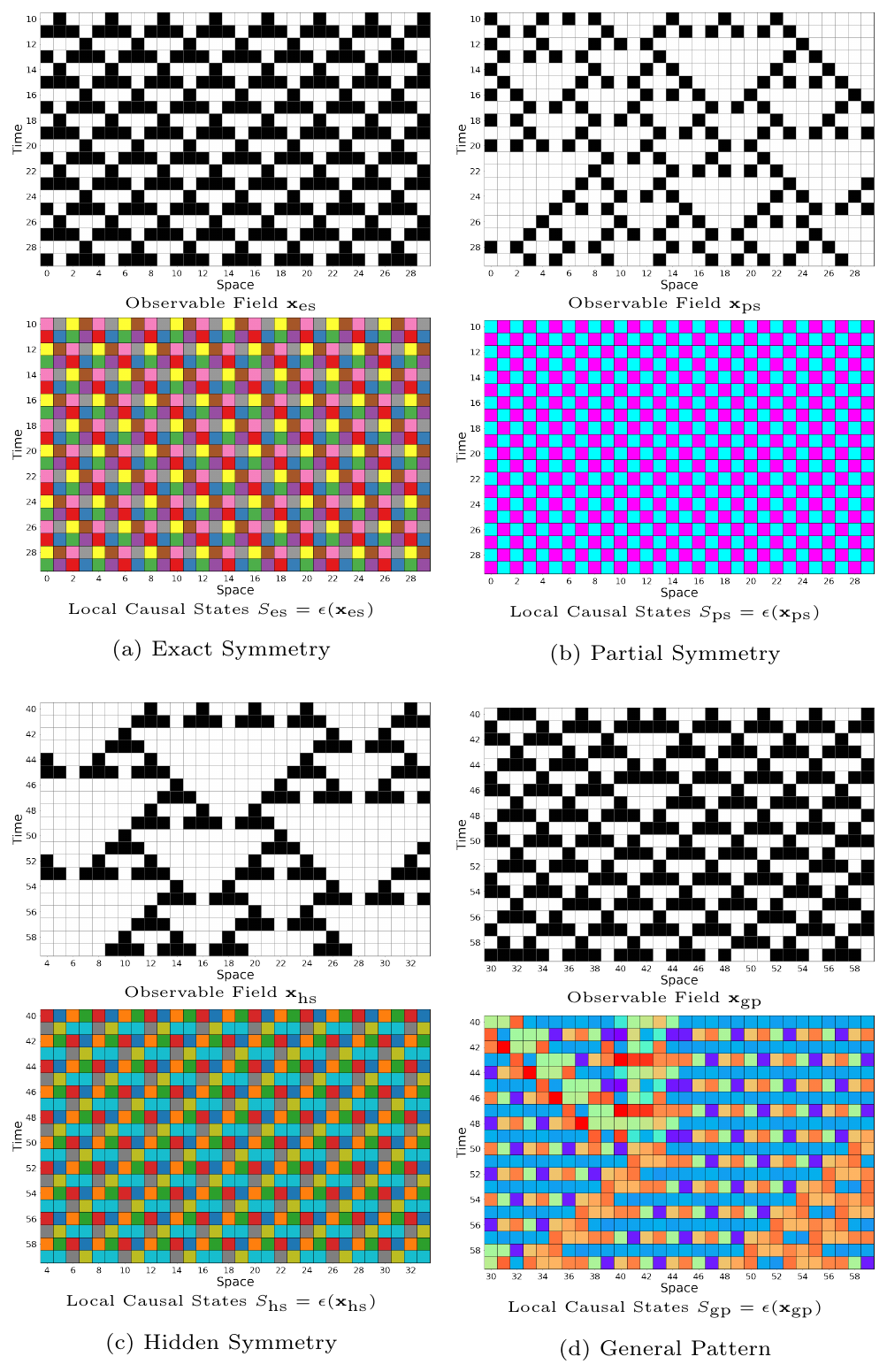}
\caption{Spacetime pattern classes: Spacetime fields $\stfield{}{}$ (above) and
	corresponding local causal state fields $\causalfield{}{} =
	\epsilon(\stfield{}{})$ (below) for (a) an exact symmetry (ECA Rule 54
	domain), (b) a partial symmetry (ECA Rule 18 domain), (c) a hidden symmetry
	(ECA Rule 22 domain), and (d) a general pattern (ECA Rule 54 evolving
	random initial configuration).
	}
\label{fig:lcsfields}
\end{figure*}

\begin{figure*}
\centering
\includegraphics[width = 0.8 \textwidth]{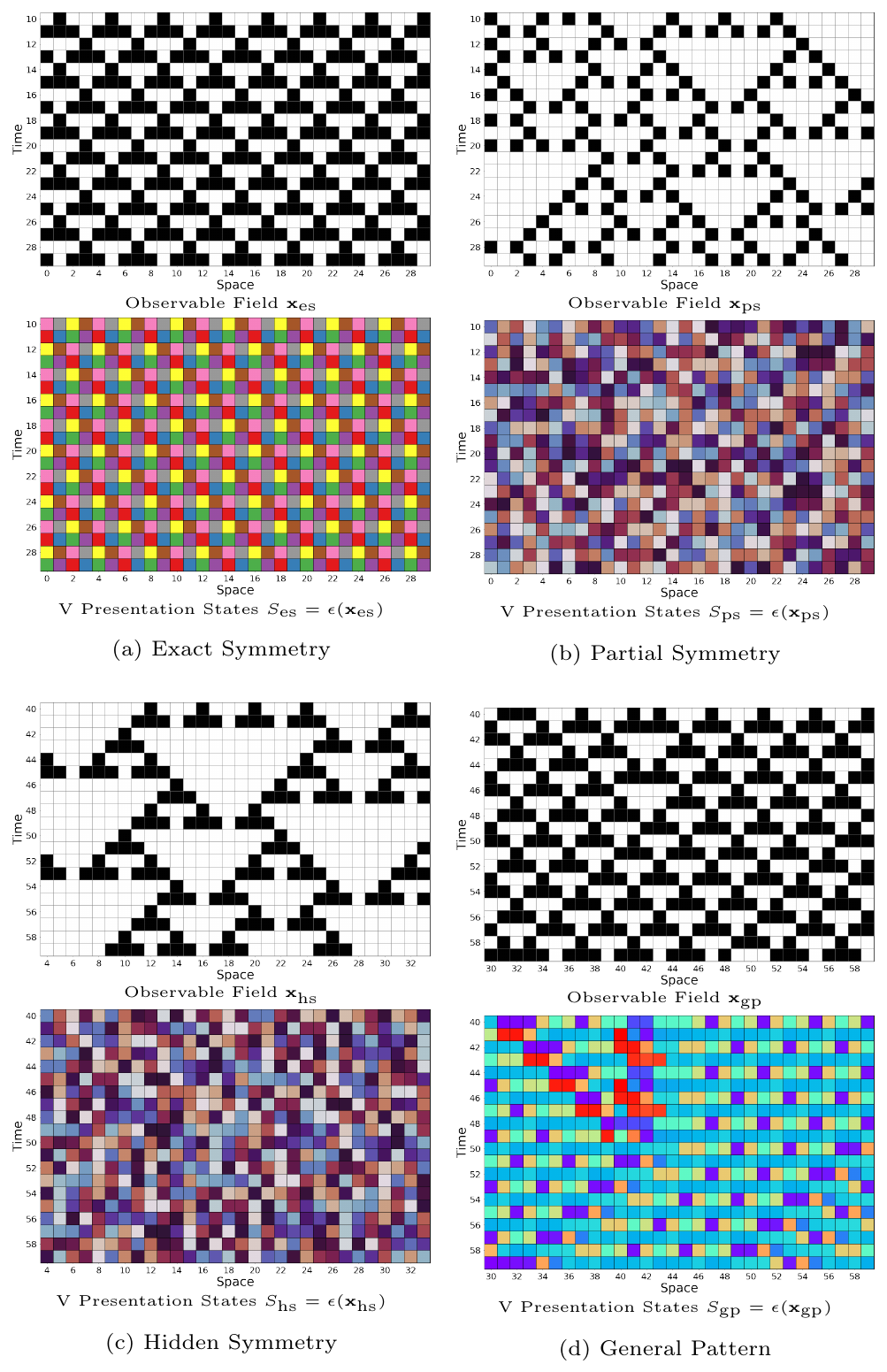}
\caption{Spacetime fields $\stfield{}{}$ shown in Fig.~\ref{fig:lcsfields}
	(above) and corresponding V state fields $\causalfield{}{} =
	\epsilon(\stfield{}{})$ (below) for (a) an exact symmetry, (b) a partial
	symmetry, (c) a hidden symmetry, and (d) a general pattern.
	}
\label{fig:Vfields}
\end{figure*}

Exact symmetries are straightforward, as just described. There is a finite $s$
and $\tau$ such that for \emph{every} site $\point$ in spacetime we have that
$\sigma_\tau\stpoint = \stpoint$ and $\sigma^s\stpoint = \stpoint$. Due to the
shared coordinate geometry and isometry equivariance of past lightcones, this
applies for the local causal state field as well: $\sigma_\tau\causalpoint =
\causalpoint$ and $\sigma^s\causalpoint = \causalpoint$. The exact symmetry
field shown in Fig.~\ref{fig:lcsfields}~(a) is a sample of the domain of ECA
Rule 54. As can be seen, both the spacetime fields and corresponding local
state fields have a period-$4$ translation symmetry in both time and space. 

For partial symmetries, some, but not all, of the spacetime coordinates return
to themselves after fixed translation, just as in one dimension. For example,
there are some sites $(r^*, t^*)$ such that $\stfield{t*}{r* + s} = \sigma^s
\stfield{t*}{r*} = \stfield{t*}{r*}$, but this does not hold for all $\point$.
As in one-dimension, a special case of partial symmetries are stochastic
symmetries that become symmetric after a wildcard substitution. For example,
consider a 1+1 dimensional spacetime field with $\Alphabet = \{0, 1\}$ such
that it has a checkerboard layout with $0$s on the black squares and wildcards
(either $0$ or $1$) on the white squares. The partial (stochastic) symmetry
example shown in Fig. \ref{fig:lcsfields} (b) is ECA Rule 18's domain, which is
a strict subset (subshift) of the $0$-wildcard checkerboard shift space. We
examine these two spacetime shift spaces in more detail below.

Unlike in one-dimension, hidden symmetries in higher dimensions correspond to
exact symmetries in the local causal-state field for spacetime fields that
themselves have no symmetries, exact or partial. The hidden symmetry spacetime
fields in Figs. \ref{fig:lcsfields} (c) and \ref{fig:Vfields} (c) are samples
of ECA Rule 22's domain and their structure is manifestly harder to detect from
visual inspection.

What is most important to note here is that the observable field $\stfield{}{}$
does not have any space or time translation symmetries, exact or partial. The
corresponding local causal-state field $\causalfield{}{} =
\epsilon(\stfield{}{})$ (shown in Fig.~\ref{fig:lcsfields}), though, does have
symmetries that are period-$4$ in both time and space. Essentially, there are
motifs that appear in $\stfield{}{}$---such as, the black ``triangles'' that
occur with a local period-$4$ structure (e.g., $1110$ in space and $1100$ in
time). However, in contrast to a stochastic symmetry field, there is no global
symmetry that captures how these stochastic motifs occur. The global hidden
symmetry is more complicated than can be revealed by simple wildcard
substitutions. It is uncovered, though, by the local causal-state field. This
structure, and its relation to the local causal states, is thoroughly detailed
in Ref. \cite{Rupe18b}. 

A general spacetime pattern then is one for which the local causal-state field
does not exhibit spacetime symmetries. There is still an algebraic relation among the local states from the spacetime shift operators, but they do not correspond to spacetime symmetries.

Note that the example in Fig.~\ref{fig:lcsfields} (d) has regions that are
\emph{locally} symmetric in both $\stfield{}{}$ and $\causalfield{}{}$.
However, this symmetry is globally broken by localized defects or coherent
structures \cite{Rupe18a}. Figure~\ref{fig:lcsfields} (d) is produced from from
ECA Rule 54 evolving a random initial configuration. The local symmetry regions
are instances Rule 54's domain---the exact symmetry field shown in
Fig.~\ref{fig:lcsfields} (a).

We emphasize that, with the exception of exact symmetries, these generalized
spacetime symmetries are empirically observed \emph{only} with local
causal-state presentations that employ future lightcone shapes. For example,
the corresponding local state fields of presentations using $\V$ local futures
is shown in Fig.~\ref{fig:Vfields} with the same spacetime fields as
Fig.~\ref{fig:lcsfields}. As expected, the exact symmetry case also has
space-time translation symmetries in the local state field in (a). However,
there are no space-time symmetries present in the local state fields for
partial (b) or hidden (c) symmetry spacetime fields when $\V$ local futures are
employed.

Altogether these observations portray a somewhat unsatisfactory scenario. On the
one hand, lightcone local futures and the resulting local causal-state models
produce insightful and predictive generalized symmetries in spacetime.
Moreover, these clearly connect to the one-dimensional sofic shifts of CA
spatial configurations, as demonstrated in Refs. \cite{Rupe18a,Rupe18b}.
However, we just demonstrated that local causal-state models are unfaithful
generators of their underlying CA shift spaces.

On the other hand, we also showed how to create faithful local generators using
$\V$ future shapes. These presentations thus appear to be the most natural
local generalization of the canonical machine presentations in one dimension.
In this case, though, the useful spacetime generalized symmetries are lost with
these generative models. 

Let us now examine in more detail the trade-offs between these two local
presentations by diving deeper into the case of stochastic symmetries. 

\section{Case Study: Stochastic Symmetries}
\label{sec:stochsym}

The partial (stochastic) symmetry spacetime fields in Figs. \ref{fig:lcsfields}
(b) and \ref{fig:Vfields} (b) are generated from ECA Rule 18 evolving a string
from its domain's invariant sofic shift. Appendix~\ref{app:domain_machines}
shows that these are points in the stochastic symmetry sofic shift with the
form $0$-$\Sigma$, where $\Sigma$ is a wildcard that can be either $0$ or $1$.

It is easy to see from Rule 18's lookup table $\phi_{18}$ that the wildcard
locations oscillate each time step between even- and odd-indexed lattice sites.
Thus, the spacetime field can be interpreted as a checkerboard pattern of fixed
$0$s and wildcards $\Sigma = \{0,1\}$; giving the checkerboard pattern seen in
the local causal-state field shown in Fig.~\ref{fig:lcsfields} (b). In this
way, the two local causal states can be interpreted as a ``fixed'' $0$ state
and a wildcard state. The local states of the $\V$ machine presentation,
though, do not occur in a checkerboard pattern in Fig.~\ref{fig:Vfields} (b).
And so, they clearly cannot be assigned such semantic labels. 

While the fixed-$0$ and wildcard semantic labels are appealing for the two
local causal states, they are also misleading. For a given spatial
configuration, the fixed-$0$ and wildcard semantics are appropriate, as it
describes the (invariant) one-dimensional sofic shift of the spatial
configurations. However, it is again easy to see from the Rule 18 lookup table
that the wildcard semantics can no longer be assigned to the full spacetime
field. Specifically, in spacetime the local update rule $\phi_{18}$ forbids
certain spacetime patches of the form:
\begin{align*}
\begin{matrix}
     \Sigma & 0 & \Sigma \\
       & \Sigma &   
\end{matrix}
\end{align*}
Therefore, the spacetime shift space of Rule 18's domain is a proper subset
(subshift) of the $0$-$\Sigma$ checkerboard shift space for which all
realizations of the above spacetime patch are allowed. 

Figure~\ref{fig:0Wfields} demonstrates that both the local causal states and
the $\V$ machine presentation local states reveal a checkerboard symmetry for
spacetime fields in the $0$-$\Sigma$ shift space. In this case, \emph{both}
types of local states certainly \emph{do} carry the semantics of fixed-$0$ and
wildcard $\Sigma$ in spacetime.

\begin{figure*}
\centering
\includegraphics[width = 0.8 \textwidth]{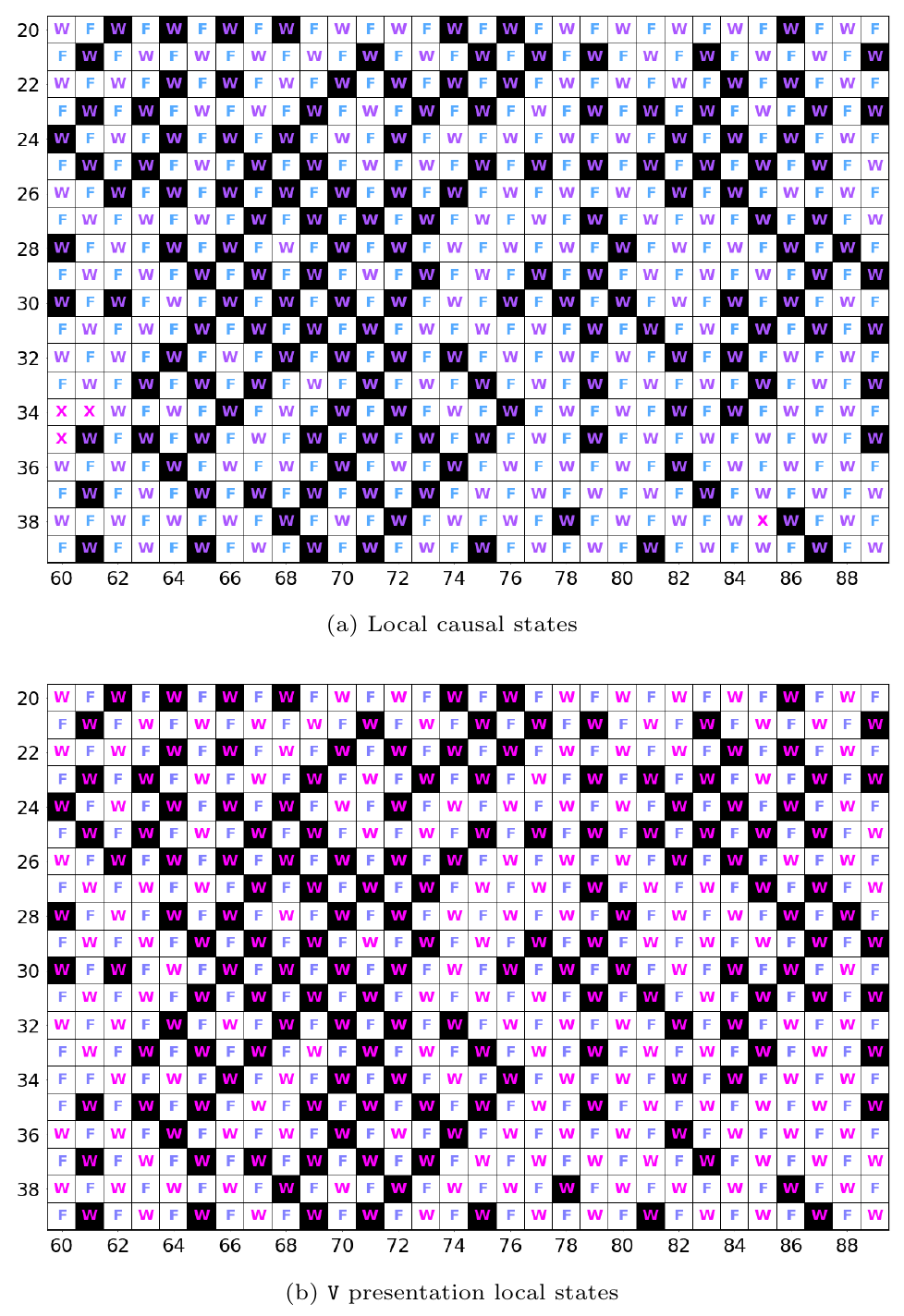}
\caption{Sample field from $0$-Wildcard shift space in black (1) and white (0)
	squares with (a) local causal states and (b) $\V$ presentation local states
	(b) overlaid. In both cases, the local states can be assigned fixed-$0$ and
	wildcard semantics. And so, they are labeled $\mathsf{F}$ and $\mathsf{W}$,
	respectively. The local causal states in (a) have an additional
	``indeterminate'' state assigned to the all-$0$ past lightcone, labeled as
	$\mathsf{X}$; see, for example, the field at $r=60$ and $t=34$). 
	}
\label{fig:0Wfields}
\end{figure*}

What does this say then about Rule 18's domain in the context of the
alternative local-state presentations in Figs. \ref{fig:lcsfields} (b) and
\ref{fig:Vfields} (b)?

First, it is interesting and not entirely clear why the choice of lightcone
local futures produces local causal states that carry the strictly-spatial
semantics of fixed-$0$ and wildcard $\Sigma$. Especially so, when we know the
wildcard semantics do not hold in spacetime. Whatever the reason, it is key to
interpreting spacetime patterns through the local causal-state fields,
including identifying coherent structures as locally-broken generalized
symmetries in spacetime \cite{Rupe18a}. 

Second, this implies that the $\V$ machine presentation states are a more
nuanced representation of the spacetime pattern of Rule 18's domain. From
above, we know the non-fixed-$0$ sites cannot be strictly interpreted as
wildcards in spacetime. However, we can interpret these as
contextually-constrained wildcards. From the $0$-$\Sigma$ patch above, the
outcome of the bottom $\Sigma$ is constrained by $\phi_{18}$ and the outcomes
of the preceding $\Sigma$s. Propagating these constraints through space and
time clearly becomes complicated very quickly. However, this is a more
appropriate semantic interpretation of Rule 18's domain. And so, this seems to
be to what the $\V$ machine states correspond. As shown in a supplementary
Jupyter Notebook \cite{genpatnb}, the local state field displayed in
Fig.~\ref{fig:Vfields} (b) is reconstructed from past lightcones of depth-$8$
and $\V$s of depth-$3$. This produces $767$ local $\V$ states to
(approximately) capture the spacetime pattern of Rule 18's domain. This is in
contrast to the three local causal states reconstructed with the same past and
future depths. (For finite-depth past lightcones, there is a third
``indeterminate'' state for the all-$0$ past lightcone.) 

Note that the fixed-$0$ semantics still holds in spacetime and that this is not
captured by the $\V$ presentation states. Given that $\V$ machine presentations
are constructed to be faithful local generators of their spacetime shift
spaces, it seems there is a similar spacetime contextuality to the fixed-$0$
sites that is necessary for faithful local generation. The spacetime
contextuality present in the $\V$ presentation states is absent in the local
causal states. And, it is the latter that reveals the spacetime symmetries
observed in Fig.~\ref{fig:lcsfields}.  

Taken all together, the spacetime patterns of partial and hidden symmetry ECA
spacetime shift spaces are exceedingly complicated. Local causal-state
presentations constructed from future lightcones and the $\V$ machine
presentations capture different aspects of these patterns. One the one hand,
$\V$ machine presentations are faithful generators and predictors of the
spacetime patterns and so capture the patterns in a more direct connection with
the one-dimensional case. However, local causal states capture generalized
spacetime symmetries that, in the case of ECA domain shift spaces, closely
connect with the (evolution of) the one-dimensional shift spaces of the
invariant spatial configurations that, in turn, possess the corresponding
generalized symmetries.



\section{Conclusions}
Building on Refs.~\cite{Crut88a} and \cite{Shal98a}'s foundations, there is a
growing body of results that use machine presentations defined from future or
predictive equivalence to discover inherent, often hidden, pattern and
structure. This includes pattern and structure in disordered crystals
\cite{Varn14a}, thermodynamic environments \cite{Boyd15a}, atmospheric
turbulence \cite{Palm00a}, bacteria behavior \cite{marz18a}, controlled quantum
processes \cite{vene20a}, and more.

The first half of the preceding development synthesized the arguments for sofic
shifts and their machine presentations as mathematical formulations of pattern
and structure. In particular, the manner in which they generalize the perfect
regularity of exact symmetries and the associated group algebra has been
rigorously clarified. It also connected to recent results on sofic measures and
their relation to stochastic \eM\ presentations of statistical patterns
supported on sofic shifts.

The development's second half overviewed the local approach to spacetime
machine presentations. We showed that the standard local causal-state approach
that uses lightcones as local futures reveals useful generalized spacetime
symmetries. However, this comes at a cost: it does not lead to faithful
generative models of the underlying shift space. This motivated introducing an
alternative local presentation model using spacetime $\V$ shapes as local
futures. We showed that these presentations are faithful generative models and
that they do not possess the same generalized spacetime symmetries of the local
causal states. The seeming mutual-exclusion of multiple local presentations
is emblematic of the difficulty of high-dimensional shift spaces. The novel
constructions given here---in particular the generative local presentations
using $\V$ futures---may provide new paths of inquiry for investigating the
organization of these rich and challenging spaces. 

One advantage of the nongenerative approach using local predictive equivalence
over future lightcones is that it does not rely on a finite alphabet for
labeled transitions to provide the algebraic structure among local causal
states. Thus, local causal states are well-defined and can be
algorithmically approximated for continuum field theories. For example, they
have been used to extract coherent structures in complex fluid flows
\cite{Rupe19a}. The results presented here provide a theoretical underpinning
and a ``physics of organization'' behind these unsupervised physics-informed machine learning
algorithms. 

Supporting Python code and supplementary Jupyter Notebooks can be found at \url{https://github.com/adamrupe/ca_patterns}.

\section*{Acknowledgments}

The authors thanks Mikhael Semaan for helpful comments and discussions. The authors also thank the Telluride Science Research Center for hospitality during
visits and the participants of the Information Engines Workshops there. 
AR acknowledges the support of the U.S. Department of Energy
through the LANL/LDRD Program and the Center for Nonlinear Studies.
JPC acknowledges the kind hospitality of the Santa Fe Institute, Institute for
Advanced Study at the University of Amsterdam, and California Institute of
Technology for their hospitality during visits. This material is based upon
work supported by, or in part by, FQXi Grant number FQXi-RFP-IPW-1902
and U.S. Army Research Laboratory and the U.S. Army Research Office under
grants W911NF-21-1-0048 and W911NF-18-1-0028.


\appendix

\section{Examples and Constructions}
\label{app:constructions}

To clarify the development and build intuition about patterns, as we use the
term, the following works through several examples in detail. Specifically, we
show explicitly how to construct a semigroup $G$ for exact and generalized
symmetries, how to construct a semiautomaton presenting $G$ (and thus
$\ShiftSpace$), and how this general semiautomaton simplifies to the future
cover, whose irreducible component is the minimal presentation
$\Presentation(\ShiftSpace)$.

\subsection{Exact Symmetry Shifts}

The simplest class of exact symmetry strings to characterize as sofic shifts
are the \emph{k-clock shifts} for which $\Alphabet = \{0, 1, \dots, (k-1)\}$
and $b = 01\cdots(k-1)$. For example, points in the $3$-clock shift are of the
form $\cdots012012012012\cdots$ with $b=012$. Intuitively, this is the simplest
case since each $a \in \Alphabet$ fully specifies the period of the translation
symmetry. Recall that this corresponds to the internal states of
$\Presentation(\ShiftSpace)$, which are the equivalence classes $[ \cdot ]_F$.
For $k$-clock sequences each internal state (excluding the forbidden state),
and thus equivalence class, is represented by a generator $a \in \Alphabet$.
This fully specifies $\ShiftSpace$ since there are no transient states. Each $a
\in \Alphabet$ is \emph{synchronizing}. In terms of the defining semigroup of
$\ShiftSpace$ for $k$-clock shift, $G$ is consists solely of $\Alphabet$ plus
the absorbing element $e$. The $3$-clock shift is given by $G = \{0, 1, 2, e\}$
with $01=1$, $12=2$, $20 = 0$, $0^2 = 1^2 = 2^2  = e$, and $02 = 10 = 21 = e$.

For more general exact symmetry shifts there are additional elements in $G$
beyond the generators $\Alphabet \cup \{e\}$. And, the future cover will have
transient states; i.e., there are additional equivalence class $[ \cdot ]_F$
beyond those in the minimal presentation $\Presentation(\ShiftSpace)$. To
illustrate, consider the shift $\ShiftSpace$ with $b=001$ and points of the
form $\cdots 001001001 \cdots$.

A simple construction of a finite $G$ for an exact symmetry shift $\ShiftSpace$
is as follows. Start by constructing the asymptotic recurrent component of the
presenting semiautomaton---this is $\Presentation(\ShiftSpace)$. We already
showed the states of $\Presentation(\ShiftSpace)$ are the equivalence classes
$[b]_F$, $[bb_1]_F$, $[bb_1b_2]_F$, and so on. These equivalence classes can generally be represented by the allowed words of length $p-1$ so that concatenation with a generator gives a (shift of) the tiling block $b$. This represents shifts of windows of length $p-1$ on points $x \in \ShiftSpace$. 

\begin{figure*}
\centering
\includegraphics[width = 0.9 \textwidth]{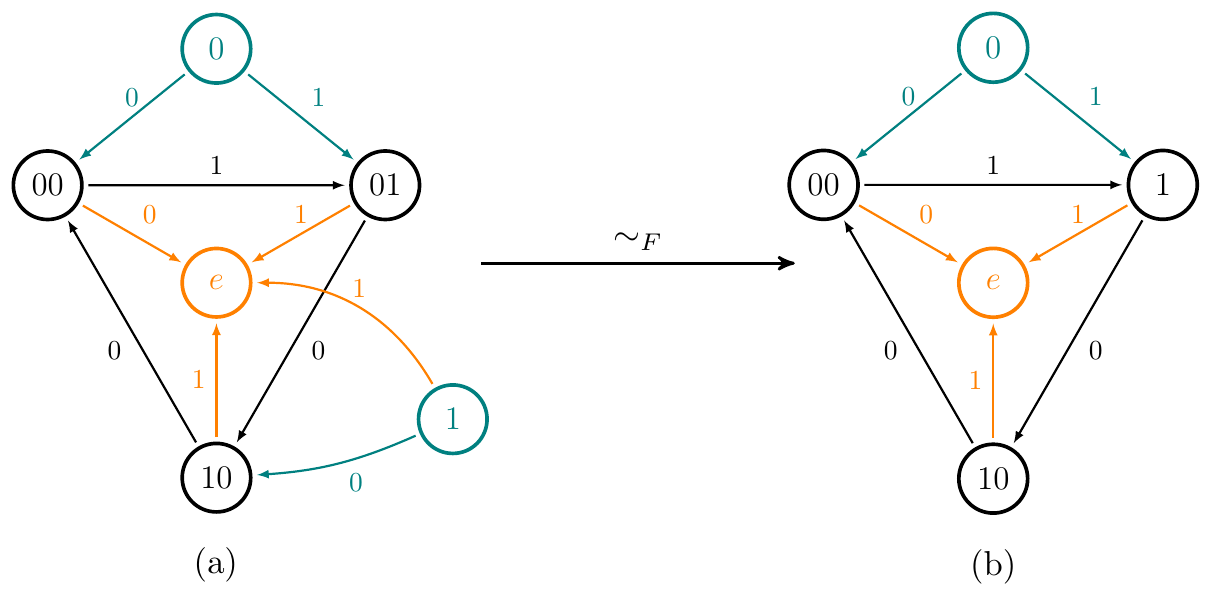}
\caption{Semiautomaton presentations for straightforward semigroup construction (a) and its simplification under the future cover equivalence relation (b).}
\label{fig:exactsym}
\end{figure*}

For our example with $b=001$, these words are $00$, $01$, and $10$. The word
$11$ is not included because it is forbidden in $\ShiftSpace$, and so
$11=1^2=e$ in $G$. If our $p-1$ window is on $00$ in $\ShiftSpace$, then a unit
shift reveals the generator $1$ and the window now shows $01$. This gives the
production rule $00 \cdot 1 = 001=01$ in $G$. In the semiautomaton presentation
there are states $\xi_{00} = [00]_F$ and $\xi_{01} = [01]_F$ with
$M_1(\xi_{00}) = \xi_{01}$. If we again shift the window on $01$, we reveal the
generator $0$ and the window now shows $10$, giving the production rule $010 =
10$. To complete the cycle we have $100 = 00$, giving the finite closure for
$G$. The rest of the elements in $G$ can be filled in with the free semigroup:
e.g., $0 \cdot 1 = 01$. Therefore, we can give a finite $G$ for our example as
$G = \{0, 1, e, 00, 01, 10\}$ with production rules $001 = 01, 010 = 10, 100 =
00$ and $ 1^2 = 0^3 = 101 = e$. The semiautomaton presentation for this $G$ is
shown in Fig.~\ref{fig:exactsym} (a). Teal colored states are transient, orange
is the absorbing forbidden state $\xi_e$, and black states are the recurrent
component (excluding $\xi_e$). Again, self-loops on $\xi_e$ are omitted for
visual clarity.

While this straightforward construct always produces a finite semigroup $G$
for a given exact symmetry shift $\ShiftSpace$, it is not necessarily a minimal
semigroup. (Or, equivalently, not a minimal presenting semiautomaton.) The
future cover equivalence relation exploits additional structure in
$\ShiftSpace$ to give a minimal description and may thus reduce and simplify
the straightforward $G$ and its presenting automaton. The $k$-clock shifts are
the extreme example, since we only need the generators as the elements of $G$
since each $a \in \Alphabet$ is synchronizing. In the $b=001$ example, while
$0$ is not synchronizing, $1$ is. Applying the future cover equivalence
relations exploits this to simplify $G$ and its presenting semiautomaton. From
visual inspection of Fig.~\ref{fig:exactsym} (a)'s semiautomaton we see that
$\xi_1$ and $\xi_{01}$ are equivalent, since their transitions lead to the same
states with the same labeled edges. Applying the future cover equivalence
relation gives the future cover, with its reduced semiautomaton presentation
shown in Fig.~\ref{fig:exactsym} (b). The simplified semigroup $G$ of the
future cover is given as $G = \{0,1, e, 00, 10\}$ with $01 = 1$ and $100=00$,
and the same production rules for forbidden words.

Since we designed our straightforward construction around the asymptotic cycle
of the translation symmetry, the recurrent components in both cases are
isomorphic and represent the canonical presentation
$\Presentation(\ShiftSpace)$ that captures the symmetry algebra.  

\subsection{General Pattern: the Even Shift}
\label{app:EvenShift}

To contrast with exact symmetry shifts, we now go through the well-studied Even
Shift \cite{Weis73,Kitc86a}. Recall that the Even Shift is the set of sequences
that have even-length blocks of $1$s bounded by $0$s. Thus, it is defined by
the set of irreducible forbidden words $\Forbidden = \{010, 01^30, 01^50,
\dots\}$. Since $\Forbidden$ is not finite, the Even Shift is not of finite
type---it is \emph{strictly sofic}. However, since it is sofic it can still be
finitely defined in terms of a finite semigroup $G$. Following Ref.
\cite{Kitc86a} we use $G = \{0, 1, e, 01, 10, 11, 101 \}$ with production rules
$0^2 = 0$, $1^3=1$, $01^2 = 1^20=0$ and $010=e$. 

The production rule $010=e$ represents the shortest forbidden word and the
other rules allow for all the other forbidden rules to reduce to $010$. For
example, we use $1^3=1$ to reduce $01110$ to $010$ which then maps to the
absorbing element $e$. We can see that while there is a countably-infinite
number of forbidden words in the Even Shift, there is structure in these
forbidden words that can be captured in a finite semigroup $G$. 

As with our explicit symmetry example $b=001$, the presenting semiautomaton for
the Even Process using the $G$ above is not minimal. Shown in
Fig.~\ref{fig:Even}, we see that there are two recurrent components
$\Presentation_A(\ShiftSpace)$ and $\Presentation_B(\ShiftSpace)$. While these
components correspond to different elements in $G$, we can see that they are
again isomorphic and so collapse together under $\sim_F$. The resulting
recurrent component is the canonical machine presentation
$\Presentation(\ShiftSpace)$, shown in Fig.~\ref{fig:machines} (c). 

\begin{figure}
\centering
\includegraphics[width = 0.5 \textwidth]{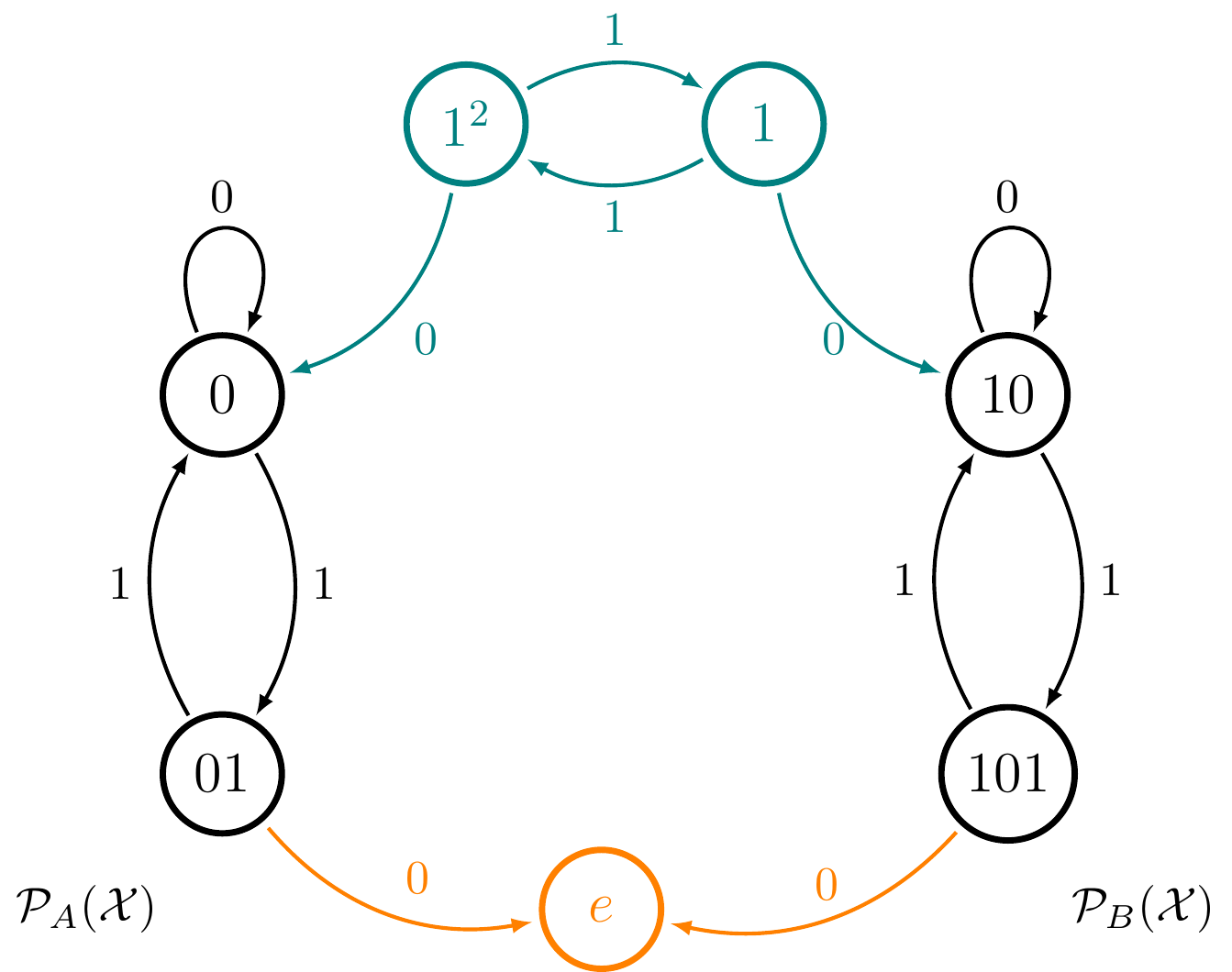}
\caption{Presenting semiautomaton for the Even Shift using
	$G = \{0,1,e,01,10,11,101\}$ and $0^2 = 0$, $1^3=1$, $01^2 = 1^20=0$,
	$010=e$. The two components $\Presentation_A(\ShiftSpace)$ and
	$\Presentation_B(\ShiftSpace)$ are isomorphic and thus collapse together
	under $\sim_F$. The recurrent component is the canonical machine
	presentation $\Presentation(\ShiftSpace)$, shown in Fig.~\ref{fig:machines}
	(c).
	}
\label{fig:Even}
\end{figure}

Though we discussed and described three forms of presenting semiautomata, we
want to emphasize the importance of the canonical presentation
$\Presentation(\ShiftSpace)$ as the mathematical representation of pattern as
generalized symmetry. There may be many semigroups $G$ that describe a given
sofic shift $\ShiftSpace$, and so there are many different semiautomata that
can present $\ShiftSpace$. However, for irreducible sofic shifts the recurrent
components of all such $G$ will be isomorphic \cite{Kitc86a}, which is why the
future cover is guaranteed to have a single recurrent component. Therefore the
future cover \emph{is} a unique representation of the structure of
$\ShiftSpace$. Patterns and the symmetries they generalize are asymptotic
properties and so the transients of the future cover are not of
interest for our current purposes. This is why the canonical presentation
$\Presentation(\ShiftSpace)$---the recurrent component of the future cover---is
the unique mathematical representation of patterns. 

Recall that the symmetry group of a translation symmetry is captured by the
canonical presentation, whose semiautomaton is a circle graph, neglecting the
absorbing state and its transitions. Similarly, we can see the canonical
presentation of the Even Shift captures the essential details of the pattern.
From inspecting its edge-labeled graph, shown in Fig.~\ref{fig:machines} (c),
we can see that an arbitrary number of $0$s are allowed from state $\xi_A$, but
once a $1$ occurs it must be followed by another $1$, ensuring the pattern of
an even number of $1$s bounded by $0$s. This also highlights the partial
regularity that motivated our definition of patterns. Starting in state $\xi_A$
there is a coin flip, either $0$ or $1$ may occur. If it is a $0$, repeat.
However, if it is a $1$, there is now additional regularity and structure that
enforces another $1$ to follow. 

\section{Distinct Statistical Patterns Supported on the Same Sofic Shift Topological Pattern}
\label{app:statpatterns}

The following illustrates how statistical patterns, in the form of sofic
measures, are additional structure on top of topological patterns. In
particular, there can be complex statistical structure supported on simple
topological structure. The topological pattern in this example is the full-$2$
shift---the set of all binary strings. As described above, the full-$2$ shift
represents a ``null'' pattern, in the sense that there is no predictability
leveraged from knowing the pattern. This is captured quantitatively by the
single-state canonical machine presentation $\Presentation(\ShiftSpace)$. Its
memory---the log of the number of states---is zero. 

The simplest statistical patterns supported on the full-$2$ shift is given by
assigning IID single-symbol probabilities, e.g., $\Pr(1) = 0.3$ and $\Pr(0) =
0.7$. The simplicity of this statistical pattern is also captured by a
single-state $\epsilon$-machine, where the above probabilities are assigned to
the correspond transitions from the state back to itself on the given symbol.
In this example it is easy to see the statistical pattern is supported on the
full-$2$ shift. Simply remove the probabilities from the single-state
transitions and the single-state canonical machine presentation of the full-$2$
shift is recovered. Since the \eM\ (statistical) and canonical machine
presentation (topological) have the same number of the states, the statistical
and topological patterns they capture can be thought of as comparable. 

\begin{figure}
\centering
\includegraphics[width = 0.45 \textwidth]{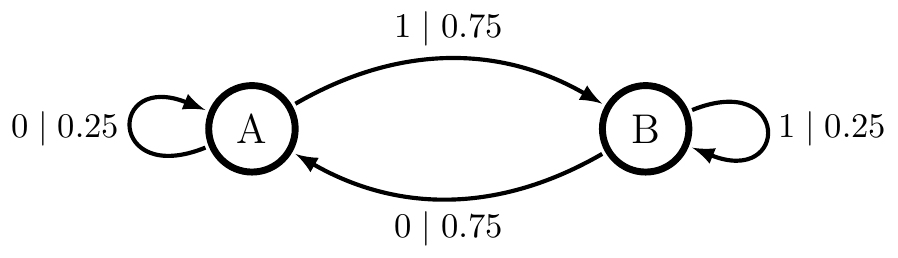}
\caption{Two state stochastic \eM\ presentation for a statistical pattern
	supported on the full-$2$ shift.
	}
\label{fig:twostatefull}
\end{figure}

However, consider the \eM\ shown in Fig.~\ref{fig:twostatefull}. Its two states
signify a more complex statistical pattern with nonzero memory. From the
machine diagram, we can see that the single-symbol probabilities depend on
which of the two internal causal states $A$ and $B$ the process is in. For
example, the probability of seeing a $1$ is $0.75$ if in causal state $A$ and
$0.25$ if in $B$. Note also that the symbol $0$ always leads to causal state
$A$ and $1$ always leads to $B$. Thus, single symbols are \emph{synchronizing}
in this example. This stochastic process is an order-$1$ Markov process. The
simple one-state case above is an IID (order-$0$ Markov) process, by contrast. 

It is perhaps clear from Fig.~\ref{fig:twostatefull}'s machine diagram
that this statistical pattern is supported on the full-$2$ shift, since both
symbols have positive probability from each of the two states. However, it is
instructive to show this using the symbol-labeled transition matrices. Recall
that $T^a_{ij}$ gives the probability of transitioning from state $\xi_i$ to
$\xi_j$ on the symbol $a \in \Alphabet$. For
Fig.~\ref{fig:twostatefull}'s \eM\ we have:
\begin{align*}
    T^0 = 
    \begin{pmatrix}
    0.25 & 0.0 \\
    0.75 & 0.0
    \end{pmatrix}
    ~,
\end{align*}
and 
\begin{align*}
    T^1 = 
    \begin{pmatrix}
    0.0 & 0.75 \\
    0.0 & 0.25
    \end{pmatrix}
    ~,
\end{align*}
where state $A$ is given index $1$ and $B$ is index $2$, so that $T^a_{12}$ is the probability of transitioning from $A$ to $B$. 

Matrix representations $M^a$ of the topological transition maps are given by
setting non-zero elements of $T^a$ to unity. In this example, we have:
\begin{align*}
    M^0 = 
    \begin{pmatrix}
    1 & 0 \\
    1 & 0
    \end{pmatrix}
    ~,
\end{align*}
and 
\begin{align*}
    M^1 = 
    \begin{pmatrix}
    0 & 1 \\
    0 & 1
    \end{pmatrix}
    ~.
\end{align*}
Recall that $M^a_{ij} = 1$ signifies a transition from internal state $i$ to
state $j$ is allowed on the symbol $a$ and $M^a_{ij} = 0$ is a forbidden
transition that produces a forbidden word. 

To see that these topological transitions correspond to the full-$2$ shift, note that the action of $M^a$ on both internal states yields the same output state, for both $M^0$ and $M^1$. That is, for both states, call them also $A= (1 \; 0)$ and $B = (0 \; 1)$, a transition on a $0$ goes to state $A$ and a transition on a $1$ goes to $B$. For example: 
\begin{align*}
    (0 \; 1) 
    \begin{pmatrix}
    1 & 0 \\
    1 & 0
    \end{pmatrix}
    = (1 \; 0)
    ~,
\end{align*}
and 
\begin{align*}
    (1 \; 0) 
    \begin{pmatrix}
    1 & 0 \\
    1 & 0
    \end{pmatrix}
    = (1 \; 0)
    ~.
\end{align*}
Thus, as described above, the two states are topologically equivalent $A \sim_F
B$ and so reduce to the single-state canonical machine presentation of the
full-$2$ shift.

Finally, as shown in Fig.~\ref{fig:fourstatefull} this construction of
probabilistically-distinct causal states supported on the full-$2$ shift can be
extended to \eM\ with a larger number of states. Clearly, the construction can
be extended indefinitely with an arbitrary number of causal states $\xi_i$ such
that $\Pr(0 | \xi_i) = p$ and $\Pr(1 | \xi_i) = 1-p$, since there can be
uncountably-many $p$. This shows that there can be arbitrarily complex
statistical patterns supported on simple topological patterns. 

\begin{figure}
\centering
\includegraphics[width = 0.45 \textwidth]{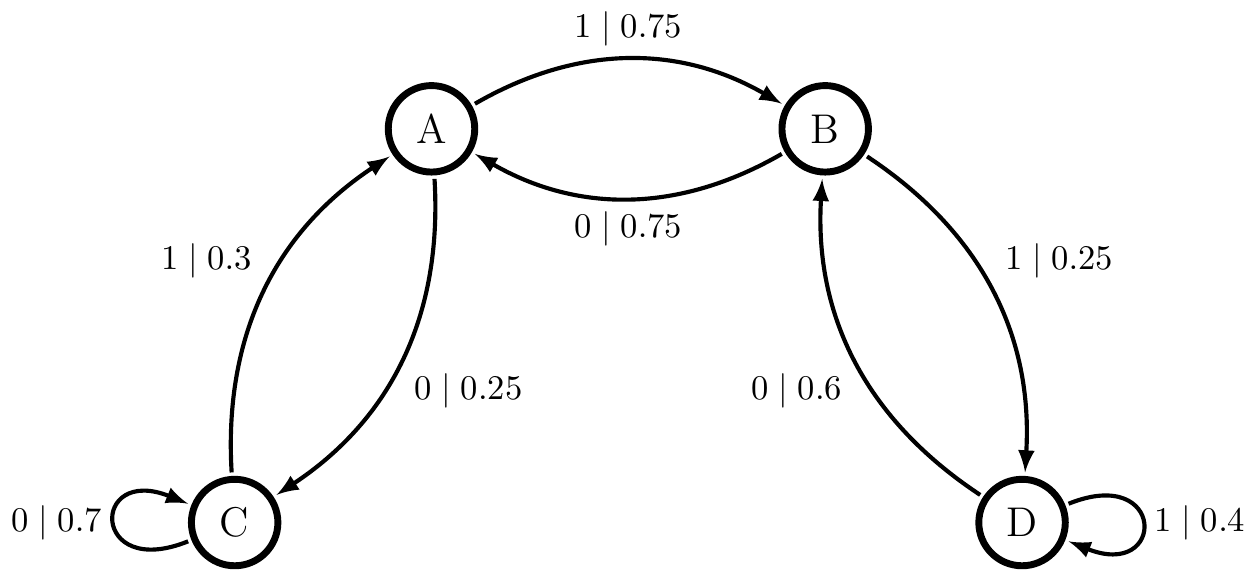}
\caption{Four-state stochastic $\epsilon$-machine presentation for a
	statistical pattern supported on the full-$2$ shift.
	}
\label{fig:fourstatefull}
\end{figure}

\section{Cellular automata}
\label{app:CA}

A one-dimensional \emph{cellular automaton} or CA $(\MeasAlphabet^{\lattice},
\Phi)$ consists of a spatial lattice $\lattice  = \mathbb{Z}$ whose
\emph{sites} take values from a finite \emph{alphabet} $\MeasAlphabet$. A CA
\emph{state} $\state  \in \MeasAlphabet^{\mathbb{Z}}$ is the configuration of
all site values $\state^r \in \MeasAlphabet$ on the lattice. (For states
$\state$, subscripts denote time; superscripts sites.) CA states evolve in
discrete time steps according to the \emph{global evolution} $\Phi: \FullShift
\rightarrow \ShiftSpace \subseteq \FullShift$, where:
\begin{align*}
\state_{t+1} = \Phi(\state_t)
~.
\end{align*}
$\Phi$ is implemented through parallel, synchronous application of a \emph{local update rule} $\phi$ that evolves individual sites $\state_t^r$ based on their radius $\radius$ \emph{neighborhoods} $\neighborhood(\state^r) =$\\$ \{ \state^{r'} \; : \;  |r-r'| < \radius \}$:
\begin{align*}
\state_{t+1}^r = \phi \big(\; \neighborhood(\state_t^r)\; \big)
  ~.
\end{align*}

Stacking the states in a CA \emph{orbit} $\state_{0:t} = \{\state_0, \state_1,
\ldots, \state_{t-1}\}$ in time-order produces a \emph{spacetime field}
$\stfield{0:t}{} \in \MeasAlphabet^{\mathbb{Z} \otimes \mathbb{Z}}$.
Visualizing CA orbits as spacetime fields reveals the fascinating patterns and
localized structures that CAs produce and how the patterns and structures
evolve and interact over time.

\subsection{Elementary cellular automata}

The parameters $(\MeasAlphabet, R)$ define a CA \emph{class}. One simple but
nontrivial class is that of the so-called \emph{elementary} cellular automata
(ECAs) \cite{Wolf83} with a binary local alphabet $\MeasAlphabet = \{0,1\}$ and
radius $\radius = 1$ local-interaction neighborhood $\neighborhood(\state_t^r) =
\state_t^{r-1} \state_t^r \state_t ^{r+1}$. Due to their definitional
simplicity and wide study, we explore ECAs in our examples.

A local update rule $\phi$ is generally specified through a \emph{lookup table}
that enumerates all possible neighborhood configurations $\neighborhood$ and
their outputs $\phi( \neighborhood)$. The lookup table for ECAs is given as:
\begin{align*}
\begin{tabular}{c c c | c}
\multicolumn{3}{c|}{$\eta$} & $O_\eta = \phi(\eta)$ \\
\hline
1 & 1 & 1 & $O_7$\\
1 & 1 & 0 & $O_6$ \\
1 & 0 & 1 & $O_5$ \\
1 & 0 & 0 & $O_4$ \\
0 & 1 & 1 & $O_3$ \\
0 & 1 & 0 & $O_2$ \\
0 & 0 & 1 & $O_1$ \\
0 & 0 & 0 & $O_0$
\end{tabular}
  ~,
\end{align*}
where each output $O_\eta = \phi(\eta) \in \MeasAlphabet$ and the $\eta$s are
listed in lexicographical order. There are $2^8 = 256$ possible ECA lookup
tables, as specified by the possible strings of output bits: $O_7 O_6 O_5 O_4
O_3 O_2 O_1 O_0$. A specific ECA lookup table is often referred to as an ECA
\emph{rule} with a \emph{rule number} given as the binary integer $o_7 o_6 o_5
o_4 o_3 o_2 o_1 o_0 \in [0,255]$. For example, ECA 172's lookup table has
output bit string $10101100$.

The $n^{\mathrm{th}}$-order lookup
table $\phi^n$ maps the radius $n \cdot \radius$ neighborhood of a
site to that site's value $n$ time steps in the future. Said another way,
a spacetime site $\state_{t+n}^r$ is completely determined by the radius
$ n\cdot \radius$ neighborhood $n$ time-steps in the past according
to:
\begin{align*}
\state_{t+n}^r
  = \phi^n\big (\neighborhood^{n} (\state_t^r)\big)
  ~.
\end{align*}
The \emph{depth-$n$ past lightcone} is the collection of all $\phi^{\ell}$ for
$1 \leq \ell \leq n$, plus the present site value itself (i.e., ``depth-0'').

\section{ECA Domain Sofic Shifts}
\label{app:domain_machines}
The generalization of invariant sets from low-dimensional dynamical systems to
high-dimensional cellular automata are known as \emph{domains}
\cite{Hans90a,Crut93a,Hans95a}. These are sets of spatial configurations that
are invariant under the global CA dynamic $\Phi$. Spatial configurations of
one-dimensional cellular automata are strings of symbols from an alphabet
$\Alphabet$. And, thus, shift spaces are natural choices for sets of spatial
configurations. In fact, the Curtis-Hedlund-Lyndon theorem shows that a CA's
dynamic---a sliding-block code---naturally induces shift-invariance in the set
of images under $\Phi$ \cite{Hedl69a}. Said another way, the dynamic of a CA on
a shift space maps to a shift space. 

For a CA $\Phi$, a \emph{domain} $\Lambda = \{\Lambda_1, \Lambda_2, \ldots,
\Lambda_p\}$ is a set of irreducible sofic shifts such that $\Phi(\Lambda_i)
\in \Lambda$. That is, if $x$ is a spatial configuration that is a point in one
of the sofic shifts $\Lambda_i \in \Lambda$, then the image $\Phi(x)$ is also a
point in one of the sofic shifts in $\Lambda$. As irreducible sofic shifts,
domains may possess patterns that are exact symmetries, partial symmetries,
hidden symmetries, or general patterns with no symmetries. Empirically, the
spacetime fields produced from the evolution of ECA domains possess the same
type of pattern as their invariant sofic shifts. The generalized spacetime
symmetries are revealed by the local causal states, as shown above in
Fig.~\ref{fig:lcsfields}. Here, we provide the machine presentations for the
domain sofic shifts used in each case.

\begin{figure}
\centering
\includegraphics[width = 0.45 \textwidth]{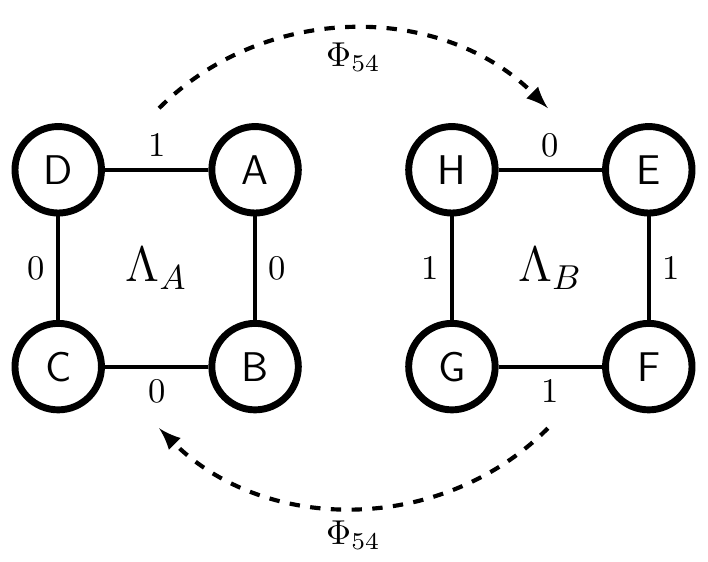}
\caption{Presenting semiautomaton for the domain of ECA Rule 54.}
\label{fig:domain54}
\end{figure}

\begin{figure}
\centering
\includegraphics[width = 0.45 \textwidth]{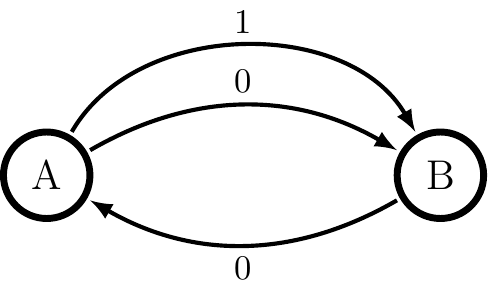}
\caption{Presenting semiautomaton for the domain of ECA Rule 18.}
\label{fig:domain18}
\end{figure}

The exact symmetry spacetime field is given by the evolution of ECA Rule 54's
domain, whose machine presentation is shown in Fig.~\ref{fig:domain54}. As can
be seen, there are two phase $\Lambda_A$ and $\Lambda_B$ that are cycled
between under $\Phi_{54}$. Each phase has an exact symmetry, with $\Lambda_{A}$
tiled by blocks of $0001$ and $\Lambda_B$ tiled by blocks of $1110$. As seen in
Fig.~\ref{fig:lcsfields} (a), this creates an exact symmetry spacetime field
that is period-$4$ in both time and space. 

The partial symmetry spacetime field possesses a stochastic symmetry and is
given by the evolution of ECA Rule 18's domain. Its machine presentation is
shown in Fig.~\ref{fig:domain18}. This is a single stochastic symmetry sofic
shift with a period-$2$ tiling of $0\Sigma$, where again $\Sigma$ is a wildcard
that can be either $0$ or $1$. Spacetime fields generated by the domain of Rule
18, as shown in Fig.~\ref{fig:lcsfields} (b), form a subshift of the
$0$-$\Sigma$ checkerboard spacetime shift space. Both are discussed in detail
above in the stochastic symmetry Case Study of Section~\ref{sec:stochsym}. 

\begin{figure}
\centering
\includegraphics[width = 0.45 \textwidth]{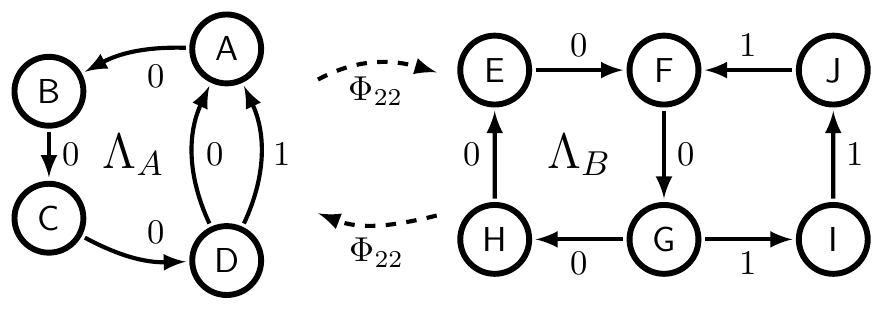}
\caption{Presenting semiautomaton for the domain of ECA Rule 22.}
\label{fig:domain22}
\end{figure}

Finally, the hidden symmetry spacetime field in Fig.~\ref{fig:lcsfields} (c) is generated by the hidden symmetry domain sofic shift of ECA Rule 22. The presenting automaton is shown in Fig.~\ref{fig:domain22}. Like Rule 54's domain, the domain of ECA Rule 22 comes in two phases. As can be seen, phase $\Lambda_{A}$ is actually a stochastic symmetry, as each state in $\Presentation(\Lambda_{A})$ returns to itself after four translations. However, this is not the case for phase $\Lambda_B$, which is a hidden symmetry sofic shift. This is because only states $\mathsf{F}$ and $\mathsf{G}$ of $\Presentation(\Lambda_B)$ return to themselves after four translations. As drawn, one can imagine ``folding'' $\Presentation(\Lambda_B)$ up onto itself to create a period-$4$ symmetry in the states. Said another way, one can think of a ``length-$3$ wildcard'' that produces either three $0$s or three $1$s. The $\Lambda_B$ sofic shift is then given by tilings of this length-$3$ wildcard followed by a fixed $0$. This is the nature of hidden symmetries. Blocks of three $0$s and three $1$s are forced to occur, followed by a fixed $0$, but the sequence of these blocks is not constrained. For example, the string of all $0$s and the string of all $1$s are both in $\Lambda_B$.

\bibliography{chaos,spacetime}

\end{document}